\documentclass[a4paper,10pt]{article}
\usepackage[english]{babel}
\usepackage[T1]{fontenc}
\usepackage[utf8]{inputenc}
\usepackage{amssymb}
\usepackage{physics}
\usepackage{amsmath}
\usepackage{graphicx}
\usepackage[unicode]{hyperref}
\usepackage{tikz}
\usepackage{color}
\usepackage{genyoungtabtikz}
\usepackage[title]{appendix}
\usepackage{indentfirst}
\usepackage{bm}
\usepackage{xhfill}
\usepackage{bbm}
\usepackage{multirow}
\usepackage{array}
\usepackage[centertableaux]{ytableau}
\usepackage{multicol}

\hypersetup{
    colorlinks=true,
    linkcolor=blue,
    filecolor=magenta,      
    urlcolor=cyan,
}
 
\usepackage{amsthm}

\theoremstyle{definition}
\newtheorem{definition}{Definition}[section]

\theoremstyle{theorem}
\newtheorem{theorem}{Theorem}[section]

\theoremstyle{proposition}
\newtheorem{proposition}{Proposition}[section]

\usepackage[nottoc]{tocbibind}

\title{\textbf{Moduli space of logarithmic states in critical massive gravities}}

\author{\textbf{Yannick Mvondo-She} \\ Department of Physics, University of Pretoria\\
Private Bag X20, Hatfield 0028, South Africa \\ \texttt{vondosh7@gmail.com}}

\date{}

\begin{document}

\maketitle

\begin{abstract}
We take new algebraic and geometric perspectives on the combinatorial results recently obtained on the partition functions of critical massive gravities conjectured to be dual to Logarithmic CFTs throught the AdS$_3$/LCFT$_2$ correspondence. We show that the partition functions of logarithmic states can be expressed in terms of Schur polynomials. Subsequently, we show that the moduli space of the logarithmic states is the symmetric product $S^n\left( \mathbb{C}^2 \right)$. As the quotient of an affine space by the symmetric group, this orbifold space is shown to be described by Hilbert series that have palindromic numerators. The palindromic properties of the Hilbert series indicate that the orbifolds are Calabi-Yau, and allow for a new interpretation of the logarithmic state spaces in critical massive gravities as Calabi-Yau singular spaces.
\end{abstract}

\tableofcontents


\section{Introduction}

Gravity in three dimensions has for some time now been an interesting model to test theories of -classical and quantum- gravity, and a consistent non-trivial theory would bring the prospect of clarifying many intricate aspects of gravity. A fundamental breakthrough was made in the study of the asymptotics, revealing the emergence of a Virasoro algebra at the boundary \cite{Brown:1986nw}. One can thus expect a
dual 2d CFT description, and this discovery can be thought of as a precursor of the AdS/CFT
correspondence. However, pure Einstein gravity in three dimensions is locally trivial at the classical level and does not exhibit propagating degrees of freedom. Hence there was a need to modify it.

One way of deforming pure Einstein gravity is by introducing a negative cosmological constant, leading to a theory with black hole solutions \cite{Banados:1992wn}. Another possibility of deformation is to add gravitational Chern-Simons term. In that case the theory is called topologically massive gravity (TMG), and contains a massive graviton \cite{Deser:1982vy} \cite{Deser:1981wh}. When both cosmological and Chern-Simons terms are included in a theory, it yields cosmological topologically massive gravity (CTMG). Such theory features both gravitons and black holes.

Following Witten's proposal in 2007 to find a CFT dual to Einstein gravity \cite{Witten:2007kt}, the graviton 1-loop partition function was calculated in \cite{Maloney:2007ud}. However, discrepancies were found in the results. In particular, the left- and right-contributions did not factorize, therefore clashing with the proposal of \cite{Witten:2007kt}. 

Soon after, a non-trivial slightly modified version of Witten's construct was proposed by Li, Song and Strominger \cite{Li:2008dq}. Their theory, in which Einstein gravity was replaced by \textit{chiral} gravity can be viewed as a special case of topologically massive gravity \cite{Deser:1982vy} \cite{Deser:1981wh}, at a specific tuning of the couplings, and is asymptotically defined with AdS$_3$ boundary conditions, according to Fefferman-Graham-Brown-Henneaux \cite{Brown:1986nw}. A particular feature of the theory was that one of the two central charges vanishes, whilst the other one can have a non-zero value. This gave an indication that the partition function could factorize.

Shortly after the proposal of \cite{Li:2008dq}, Grumiller \textit{et al.} noticed that relaxing the Brown-Henneaux boundary conditions allowed for the presence of a massive mode that forms a Jordan cell with the massless graviton, leading to a degeneracy at the critical point \cite{Grumiller:2008qz}. In addition, it was observed that the presence of the massive mode spoils the chirality of the theory, as well as its
unitarity. Based on these results, the dual CFT of critical cosmological topologically massive gravity (cTMG) was conjectured to be logarithmic, and the massive mode was called the \textit{logarithmic partner} of the graviton. Indeed, Jordan cell structures are a noticeable feature of logarithmic CFTs, that are non unitary theories (see \cite{Gurarie:1993xq} as well as the very nice introductory notes \cite{Flohr:2001zs} and \cite{Gaberdiel:2001tr}). The correspondence distinguishes itself on the conjectured dual LCFT side by a left-moving energy-momentum tensor $T$ that has a logarithmic partner state $t$ with identical conformal weight, forming the following Jordan cell

\begin{eqnarray}
 L_0   \begin{pmatrix} \ket{T} \\ \ket{t} \end{pmatrix} =  \begin{pmatrix} 2  & 0 \\ 1 & 2 \end{pmatrix} \begin{pmatrix} \ket{T} \\ \ket{t} \end{pmatrix}.
\end{eqnarray}

A major achievement towards the formulation of the correspondence was the calculation of correlation functions in TMG \cite{Skenderis:2009nt,Grumiller:2009mw}, which confirmed the existence of logarithmic correlators of the type $\langle T(x) t(y) \rangle =b_L/(x-y)^4$ that arise in LCFT, with $b_L$ commonly referred as the logarithmic anomaly. Subsequently, the 1-loop graviton partition function of cTMG  on the thermal AdS$_3$ background was calculated in \cite{Gaberdiel:2010xv}, resulting in the following expression

\begin{eqnarray}
\label{z tmg}
{Z_{\rm{cTMG}}} (q, \bar{q})= \prod_{n=2}^{\infty} \frac{1}{|1-q^n|^2} \prod_{m=2}^{\infty} \prod_{\bar{m}=0}^{\infty} \frac{1}{1-q^m \bar{q}^{\bar{m}}},
\end{eqnarray}
where the first product can be identified as the three-dimensional gravity partition function $Z_{0,1}$ in \cite{Maloney:2007ud}, and is therefore not modular invariant. The double product is the partition function of the logarithmic single and multi particle logarithmic states, and will be the central object of this work.

The corresponding expression of the partition function on the CFT side was derived in \cite{Gaberdiel:2010xv} and given the form

\begin{equation}
 {\rm{Z_{LCFT}}}(q, \bar{q})  = {\rm{Z^0_{LCFT}}} (q, \bar{q}) + \sum_{h,\bar{h}} N_{{h,\bar{h}}} q^h \bar{q}^{\bar{h}} \prod_{n=1}^{\infty} \frac{1}{|1-q^n|^2},
\end{equation}
with 
\begin{equation}
{\rm{Z^0_{LCFT}}(q, \bar{q})= Z_{\Omega} + Z_t} = \prod_{n=2}^{\infty} \frac{1}{|1-q^n|^2} \left( 1 + \frac{q^2}{|1-q|^2}\right),
\end{equation}
where $\Omega$ is the vacuum of the holomorphic sector, and $t$ denotes the logarithmic partner of the energy momentum tensor $T$.

As pointed out in \cite{Grumiller:2013at}, a better understanding of the partition function from the CFT side is desirable, in particular how to precisely match the combinatorics of multi particle logarithmic states on the gravity side to states on the CFT side. Motivated by that, the partition function  was reformulated in \cite{Mvondo-She:2018htn}, and recast in terms of Bell polynomials. Furthermore, it was shown that the partition function could be rewritten using the more usual language of Hilbert series, leading to an identity between the generating function of Bell polynomials and the celebrated plethystic exponential. In addition, the Bell polynomials formalism showed an interesting use in revealing hidden symmetry actions on the $n$-particle terms of the partition function (this point will be given an interpretation in this work).  

Despite the aforementioned achievements made since the conjecture of the AdS$_3$/LCFT$_2$ in 2008, it is fair to say that very little is known about the nature of the logarithmic states. In this paper, we try to deal with this issue by exploiting the results of \cite{Mvondo-She:2018htn} to study the moduli space of the logarithmic sector. The concept of moduli space originates from algebraic geometry. The behaviour of certain geometric objects such as collections of $n$ distinct ordered points on a given topological space can be understood by finding a space $X$ which parametrizes these objects, i.e a space whose points are in bijection with equivalence classes of these objects. A space $X$ with such a correspondence is called a moduli space, and it parametrizes the types of objects of interest, which in our case will be the logarithmic states. The geometry of moduli spaces can be encoded in their generating functions. We take advantage of that fact to give a symmetric group interpretation of the results obtained in \cite{Mvondo-She:2018htn}, and to show that as the symmetric product $S^n ( \mathbb{C}^2 )$, the moduli space of the logarithmic states described by palindromic Hilbert series is a Calabi-Yau singular space. 

This paper is organized as follows. In section 2, we give a brief description of partition functions in critical massive gravities. In particular, we recall how some infinite products can usefully be rewritten as generating functions of Bell polynomials. This is an interesting application of Bell polynomials in theoretical physics. In section 3, we discuss the symmetric product structure that appears from the sub-partition function of the logarithmic sector. We start by showing that the counting in the latter expressed in terms of Bell polynomials, is related to the cycle index of the symmetric group, $i.e$ a polynomial in several variables that counts objects that are invariant under the action of the symmetric group. A symmetric group interpretation of the counting of the logarithmic states is then given by showing that the partition function of those states is the generating function of the 1-part Schur polynomials. Then in section 4, we show that the moduli space of the logarithmic states is the $n$-th symmetric product of $\mathbb{C}^2$, by showing that the partition function of the logarithmic states is the generating function of Molien series. We also give an interpretation of some of the results from \cite{Mvondo-She:2018htn}, in which differential operators acting on the  Bell polynomials were constructed, as differential operators acting on orbifolds. In section 5, using properties of Bell polynomials, we show that the Hilbert series of the moduli space $S^n \left( \mathbb{C}^2 \right)$ have a very interesting palindromic property. Such property has already been discussed in the physics literature, for instance in the context of Hilbert series for moduli spaces of supersymmetric vacua of gauge theories \cite{Gray:2008yu,Hanany:2008sb}, or again in \cite{Forcella:2008bb,Forcella:2008eh}, and more recently in the context of primary fields in dimension four free CFT \cite{deMelloKoch:2017caf,deMelloKoch:2017dgi}. According to a beautiful theorem by Stanley \cite{stanley1978hilbert}, the palindromic property of Hilbert series associated to the counting of the logarithmic states indicates that the moduli space is Calabi-Yau. Finally, a brief discussion is given in section 6, then we conclude and give some research prospects for the future in section 7.


\section{Partition functions of critical massive gravities}
The graviton 1-loop partition function of cosmological topologically massive gravity and new massive gravity both at the critical point were calculated in \cite{Gaberdiel:2010xv}. In the case of topologically massive gravity, the computation was given the form

\begin{eqnarray}
\label{eq:1}
{Z_{\rm{cTMG}}} (q, \bar{q})= \prod_{n=2}^{\infty} \frac{1}{|1-q^n|^2} \prod_{m=2}^{\infty} \prod_{\bar{m}=0}^{\infty} \frac{1}{1-q^m \bar{q}^{\bar{m}}},
\end{eqnarray}
and for new massive gravity, the partition function was derived as

\begin{eqnarray}
\label{cNMG}
{Z_{\rm{cNMG}}}(q, \bar{q})= \prod_{n=2}^{\infty} \frac{1}{|1-q^n|^2} \prod_{m=2}^{\infty} \prod_{\bar{m} =0}^{\infty} \frac{1}{1-q^m \bar{q}^{\bar{m}}} \prod_{l=2}^{\infty} \prod_{\bar{l}=0}^{\infty} \frac{1}{1-q^l \bar{q}^{\bar{l}}}.
\end{eqnarray}

Shortly after these results, topologically massive gravity was generalized to higher spins in \cite{Bagchi:2011vr}, and the 1-loop partition function for topologically massive higher spin gravity (cTMHSG) for arbitrary spin was calculated in \cite{Bagchi:2011td}. A special attention was given to spin-3 case for which the partition function was expressed as

\begin{eqnarray}
\label{cTMHSG}
{\rm{Z^{(3)}_{cTMHSG}}} (q, \bar{q})&=& \prod_{n=2}^{\infty} \frac{1}{|1-q^n|^2} \prod_{m=2}^{\infty} \prod_{\bar{m}=0}^{\infty} \frac{1}{1-q^m \bar{q}^{\bar{m}}} \nonumber \\  &\times& \left[ \prod_{n=3}^{\infty} \frac{1}{|1-q^n|^2} \prod_{m=3}^{\infty} \prod_{\bar{m}=0}^{\infty} \frac{1}{1-q^m \bar{q}^{\bar{m}}} \prod_{m=4}^{\infty} \prod_{\bar{m}=3}^{\infty} \frac{1}{1-q^m \bar{q}^{\bar{m}}} \right]. \nonumber \\
\end{eqnarray}

Recently, motivated by the desire for a better understanding of the combinatorics of
the logarithmic excitations in the partition function of these critical massive gravities, and with the eventual goal of having a more concrete grasp of their conjectured holographic (L)CFT duals, the partition functions were shown to be usefully expressed in terms of Bell polynomials \cite{Mvondo-She:2018htn}.

Bell polynomials are very useful in many areas of mathematics and have enjoyed many applications in physics as well. For instance, recently expressions of canonical and grand canonical partition functions of interacting quantum gases of Statistical Mechanics systems were rederived in terms of Bell Polynomials by the authors of \cite{zhou2018canonical}, in which the Bell polynomials are the Mayer (cluster) expansion. Also, in \cite{fronczak2012microscopic} \cite{siudem2013partition} \cite{fronczak2013cluster} \cite{siudem2016exact} and references therein, the use of Bell polynomials is discussed for partition functions, suggesting some interaction of particles in the theories concerned.

Bell polynomials have also appeared in the study of partition functions of BPS bosonic operators. Following \cite{Lucietti:2008cv}, if we consider such partition functions at finite $N$, and denote them $Z_k( \vec{\beta} ;N)$ where $ \vec{\beta} = (\beta_1,\ldots,\beta_k)$ is a set of $k$ chemical potentials conjugate to $n_i$, the quantum numbers of the various conserved charges in the superconformal theory in
question, and $N$ the rank of the gauge group, then these partition functions typically take the following expression 

\begin{eqnarray}
\mathcal{Z}_k( \vec{\beta} ,p)= \prod_{n_1,n_2,\ldots,n_k \geq 0}^{\infty} \frac{1}{1- p \exp{(- \vec{\beta} \cdot \vec{n}) }},
\end{eqnarray}
where the infinite product converges if $|p|<1$ and $\Re{(\beta_i)}>0$. Furthermore, the infinite product is the generating function for $Z_k(\vec{\beta};N)$ expressed as

\begin{eqnarray}
\label{grand canonical PF}
\mathcal{Z}_k(\vec{\beta},p) = \sum_{N=0}^{\infty} Z_k( \vec{\beta} ;N) p^N.
\end{eqnarray}
Eq. (\ref{grand canonical PF}) is the grand canonical partition function for bosons in a $k$-dimensional harmonic oscillator potential with $p$ as the fugacity, defined as the chemical potential that keeps track of particle number $N$. These partition functions correspond for instance to $\frac{1}{2}$-BPS or $\frac{1}{4}$-BPS states in $\mathcal{N}=4$ SYM when $k=1,2$ respectively, or to $\frac{1}{8}$-BPS states in the M2-brane world-volume for $k=4$ and to $\frac{1}{4}$-BPS states in the M5-brane world-volume for $k=2$ (the $(2,0)$ SCFT in six dimensions) \cite{Lucietti:2008cv}.

In \cite{Mvondo-She:2018htn}, it was shown that the partition functions of critical gravities can be expressed in terms of Bell polynomials. In the specific case of critical cosmological TMG, writing Eq. (\ref{z tmg}) as 

\begin{eqnarray}
\label{z grav z log}
{Z_{\rm{cTMG}}} (q, \bar{q})=  {Z_{\rm{gravity}}} (q, \bar{q}) \cdot {Z_{\rm{log}}} (q, \bar{q}),
\end{eqnarray}
where

\begin{eqnarray}
{Z_{\rm{gravity}}} (q, \bar{q})= \prod_{n=2}^{\infty} \frac{1}{|1-q^n|^2}, \hspace{0.5cm} \mbox{and}
\hspace{0.5cm} {Z_{\rm{log}}} (q, \bar{q}) = \prod_{m=2}^{\infty} \prod_{\bar{m}=0}^{\infty} \frac{1}{1-q^m \bar{q}^{\bar{m}}},
\end{eqnarray}
it was shown that ${Z_{\rm{log}}} (q, \bar{q})$ is the generating function of Bell polynomials 

\begin{eqnarray}
\label{z log bell}
{Z_{\rm{log}}} (q, \bar{q}) = \sum_{n=0}^{\infty} \frac{Y_n}{n!} \left( q^{2} \right)^n. 
\end{eqnarray}
In Eq. (\ref{z log bell}), $Y_n$ is the (complete) Bell polynomial with variables $g_1, g_2, \ldots, g_n$ such that 

\begin{eqnarray}
Y_{n}(g_1, g_2, \ldots, g_{n})= \sum_{\vec{k} \vdash n} \frac{n!}{k_1! \cdots k_n!} \prod_{j=1}^{n} \left( \frac{g_j}{j!} \right)^{k_j},
\end{eqnarray}
with 

\begin{eqnarray}
\vec{k} \vdash n = \left\{ \vec{k} = \left( k_1,k_2,\ldots,k_n  \right) \hspace{0.1cm}| \hspace{0.1cm} k_1 + 2 k_2 + 3 k_3 +\cdots + n k_n = n \right\},
\end{eqnarray}
and

\begin{eqnarray}
g_n = (n-1)! \sum_{m \geq 0, \bar{m} \geq 0}  q^{nm} \bar{q}^{n \bar{m}} = (n-1)! \frac{1}{|1-q^n|^2}.
\end{eqnarray}
Similarly, Eq. (\ref{cNMG}) takes the form

\begin{subequations}
\begin{align}
{\rm{Z_{cNMG}}}(q, \bar{q}) &= {Z_{\rm{gravity}}} (q, \bar{q}) \cdot {Z_{\rm{log}}} (q, \bar{q}) \cdot  {\bar{Z}_{\rm{log}}} (q, \bar{q}) \\ &= \label{cNMG with Bell pol.}  \prod_{n=2}^{\infty} \frac{1}{|1-q^n|^2} \left( \sum^{\infty}_{m=0} \frac{Y_m}{m!} \left( q^2 \right)^m \right) 
\left( \sum^{\infty}_{l=0} \frac{Y_l}{l!} \left( \bar{q}^2 \right)^l  \right), 
\end{align}
\end{subequations}

\noindent while Eq. (\ref{cTMHSG}) can be written as 

\begin{eqnarray}
{\rm{Z^{(3)}_{cTMHSG}}} (q, \bar{q}) &=& \chi_0 (\mathcal{W}_3) \times \bar{\chi}_0 (\mathcal{W}_3) \left( \sum^{\infty}_{m=0} \frac{Y_m}{m!} \left( q^2 \right)^m \right) \nonumber \\ &\times& \left( \sum^{\infty}_{l=0} \frac{Y_l}{l!} \left( q^3 \right)^l \right)
\left( \sum^{\infty}_{k=0} \frac{Y_k}{k!} \left( q^4 \bar{q}^3 \right)^k  \right),
\end{eqnarray}
with $\chi_0 (\mathcal{W}_3)$ and $\bar{\chi}_0 (\mathcal{W}_3)$ as the holomorphic and antiholomorphic $W_3$ vacuum characters.

The logarithmic partition function can therefore be given a general form that reads

\begin{eqnarray}
Z_{\rm{log}}(\nu; q,\bar{q}) = \prod_{m=0}^{\infty} \prod_{\bar{m}=0}^{\infty} \frac{1}{1- \nu q^m \bar{q}^{\bar{m}}} = \sum_{n=0}^{\infty} \frac{Y_n}{n!} \nu^{n},
\end{eqnarray}
where the variable $\nu$ represents a monomial in (highest) weight that can be holomorphic denoted by $q^h$ or antiholomorphic denoted by $\bar{q}^{\bar{h}}$, with $h$ and $\bar{h}$ as the conformal weights of holomorphic and antiholomorphic logarithmic partner states. 

Since the conjecture of critical topologically massive gravity as the dual of a logarithmic conformal field theory, it is fair to say that little work has been done in the description of the logarithmic states. In the next section, we would like to make a few steps in that direction by using results recently obtained from the partition function of critical massive gravities to extract information about the moduli space of the logarithmic states.
\section{A representation theoretic aperçu of $Z_{\rm{log}}(\nu; q,\bar{q})$}
\label{representation theory section}
In this section, we give a representation theoretic interpretation of the results obtained in \cite{Mvondo-She:2018htn}. We start by giving some preliminaries intended to give grounds for the transition from the combinatorial description summarized in the previous section to the representation theoretic language that will be used in this section.

\subsection{Preliminaries}

\subsubsection{Permutations and cycles of the Symmetric Group $S_n$}
\label{Permutations and cycles of the Symmetric Group}
Consider a finite set  denoted by $\mathbb{S}= \left\{ 1, \ldots, n \right\}$. A permutation  of $\mathbb{S}$ is a one-to-one mapping of $\mathbb{S}$ onto itself. The symmetric group $S_n$ is the group of all permutations of the $n$ elements of $\mathbb{S}$. The order $\left| S_n \right|$ of the group, i.e the number of elements of $S_n$ is equal to $n!$. An effective way of describing permutations is by using the language of \textit{cycles}. Indeed, given a permutation $\pi$, $\mathbb{S}$ can be split into cycles, which are subsets of $\mathbb{S}$ subject to cyclic permutation by $\pi$. As a result, every permutation of the elements of $\mathbb{S}$ can be written as a product of disjoint cycles. For example, the symmetric group $S_3$ consists of the $3!=6$ permutation elements

\begin{subequations}
\begin{align}
& \mathbbm{1}= \mbox{identical permutation}, \\
& 1 \rightarrow 2 \rightarrow 1 \hspace{0.5cm} \mbox{and}  \hspace{0.5cm} 3 \rightarrow 3,\\ 
& 2 \rightarrow 3 \rightarrow 2 \hspace{0.5cm} \mbox{and}  \hspace{0.5cm} 1 \rightarrow 1,\\
& 3 \rightarrow 1 \rightarrow 3 \hspace{0.5cm} \mbox{and}  \hspace{0.5cm} 2 \rightarrow 2,\\
& 1 \rightarrow 2 \rightarrow 3 \rightarrow 1, \\ 
& 1 \rightarrow 3 \rightarrow 2 \rightarrow 1, 
\end{align}
\end{subequations}

\noindent which can be expressed in cycle notation as 

\begin{eqnarray}
S_3 = \left\{ \mathbbm{1}, (12),(23),(31), (123), (132) \right\},
\end{eqnarray}

\noindent where (123) and (132) are cycles of length 3, (12),(23) and (31) are cycles of length 2, and $\mathbbm{1}$ has length 1. 

A permutation can be assigned the \textit{cycle symbol} 

\begin{eqnarray}
\label{cycle symbol}
1^{m(1)} 2^{m(2)} \cdots n^{m(n)},
\end{eqnarray}
if its disjoint cycle product contains $m(k)$ number of $k$-cycles, with $1 \leq k \leq n$. The number $m(k)$ is called the \textit{multiplicity} of $k$-cycles in the disjoint cycle product of the given permutation. For instance, the permutation 

\begin{eqnarray}
\begin{pmatrix}
1 & 2 & 3 & 4 & 5 & 6 & 7 & 8 & 9 \\
2 & 3 & 1 & 5 & 4 & 6 & 8 & 7 & 9
\end{pmatrix}
= (6)(9)(45)(78)(123),
\end{eqnarray}

\noindent has the cycle symbol $1^2 2^2 3^1$. 

\subsubsection{Partitions}
The cycle structure of group elements in $S_n$ can be represented by partitions of $n$. A partition is a sequence 

\begin{eqnarray}
\lambda = \left(  \lambda_1, \lambda_2, \ldots, \lambda_k, \ldots  \right) 
\end{eqnarray}

\noindent of non-negative integers in non-increasing order 

\begin{eqnarray}
\lambda_1 \leq \lambda_2 \leq \ldots \leq \lambda_k \leq \ldots
\end{eqnarray}

\noindent that contains a finite number of non-zero terms. $\lambda_k$ is called the \textit{parts} of $\lambda$. The \textit{length} $l(\lambda)$ is the number of parts of $\lambda$. The \textit{weight} $\left| \lambda \right|$ is the sum of the parts, and a partition with weight $|\lambda|=n$ is a partition of $n$ denoted $\lambda \vdash n $.

A partition $|\lambda|$ of $n$ can conveniently be expressed using a notation that indicates the number of times $m(k)$ an integer $k$ occurs as a part  

\begin{eqnarray}
\label{partition=cycle symbol}
\lambda = \left( 1^{m(1)} 2^{m(2)} \cdots k^{m(k)} \cdots   \right). 
\end{eqnarray}

\noindent Eq. (\ref{partition=cycle symbol}) is very similar to the cycle symbol of a permutation seen in Eq. (\ref{cycle symbol}). This is because partitions of $n$ are in one to one correspondence with the cycle structure of $S_n$. As an example, the partitions of 5 are ordered as 

\begin{eqnarray}
\label{partitions of 5}
(5), (41), (32), (3 1^2), (2^2 1), (2 1^3), (1^5).
\end{eqnarray}

\subsubsection{Young diagrams}
Partitions can be graphically represented by Young diagrams. They are are denoted by $R \vdash n$
 and consist of diagrams with $n$ boxes arranged in left-justified rows stacked in such a way that the order of the row lengths is weakly decreasing. Given this convention, in the partition $\lambda$ of $n$ the $k^{\rm{th}}$ part $\lambda_k$ corresponds to the $k^{\rm{th}}$ row of the frame, consisting of $\lambda_k$ boxes. The partitions of 5 in Eq. (\ref{partitions of 5}) can for instance be expressed in terms of Young diagrams as follows
 
 
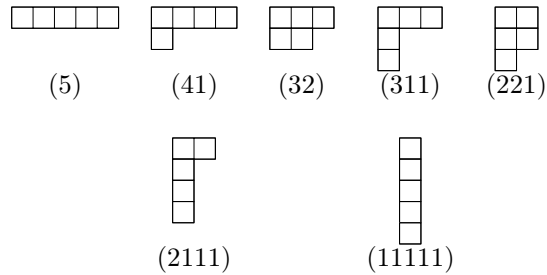
\begin{center}
\begin{tabular}{ccccc}
 $\Yboxdim{8pt} \yng(5)$ & $\Yboxdim{8pt} \yng(4,1)$ & $\Yboxdim{8pt} \yng(3,2)$ & $\Yboxdim{8pt} \yng(3,1,1)$ & $\Yboxdim{8pt} \yng(2,2,1)$\\ 
 (5) & (41) & (32) & (311) & (221) \\  
 &&&& \\
 & $\Yboxdim{8pt} \yng(2,1,1,1)$ &  & $\Yboxdim{8pt} \yng(1,1,1,1,1)$ & \\ 
 & (2111) & & (11111) & \\
\end{tabular}
\end{center}

The relationship between partitions, Young diagrams and disjoint cycles representing elements of the symmetric group can be illustrated for $n=2$ and $n=3$ in Table 1 and Table 2. 

\begin{table}[h]
\centering
\renewcommand{\arraystretch}{1.9}
\begin{tabular}{ |m{1.8cm}|c|c|c| } 
 \hline
 Partition & Young diagram & Disjoint cycle product &  \\ \hline
 2 = 2 & $\Yboxdim{8pt} \yng(2)$  & (12) & $2^1$ \\ \hline
 2 = 1 + 1 & $\Yboxdim{8pt} \yng(1,1)$ & (1)(2) & $1^2$ \\ [1.5ex] \hline
\end{tabular}
\caption{Young diagrams for the symmetric group $S_2$}
\label{table 1}
\end{table}

\begin{table}[h]
\centering
\renewcommand{\arraystretch}{1.9}
\begin{tabular}{ |m{2.5cm}|c|c|c| } 
 \hline
 Partition & Young diagram & Disjoint cycle product &  \\ \hline
 3 = 3 & $\Yboxdim{8pt} \yng(3)$ & (123), (132) & $3^1$ \\ \hline
 3 = 2 + 1 & $\Yboxdim{8pt} \yng(2,1)$ & (12)(3), (23)(1), (31)(2) & $2^1 1^1$ \\ [1.5ex] \hline
 3 = 1 + 1 + 1 & $\Yboxdim{8pt} \yng(1,1,1)$ & (1)(2)(3) & $1^1 1^1 1^1$ \\ [3.5ex]
 \hline
\end{tabular}
\caption{Young diagrams for the symmetric group $S_3$}
\label{table 2}
\end{table}

\subsubsection{Symmetry and the cycle index}
We close the preliminaries by make the first connection between the combinatorial results obtained in \cite{Mvondo-She:2018htn} with the representation theory of the Symmetric Group.

In combinatorics, only few formulae can be applied systematically in all cases of a given problem. \textit{P\'{o}lya theory} is one such example, as it enables to count the number of items under specific constraints, such as number of colors or more generally symmetry.  From a group theory perspective, counting objects such as states "up to symmetry" means counting the orbits of some group of symmetries on the set of states that are being counted. A standard procedure to solve this problem is to use the \textit{orbit-counting (Burnside's) lemma}  \cite{cameron_1994}. Alternatively, the counting can be made systematic by using a multivariate polynomial associated with a permutation group, called the \textit{cycle index}. Before discussing the latter point further, we make the following observations. Looking at Table \ref{table 1}, in the case of $S_2$, reexpressing the only cycle $2^1$ by a variable $g_2$ and the two disjoint cycle products (1)(2) with group theoretic notation $1^2$ by $g^2_1$, we see that the function $Y_2 = g_1^2 + g_2$ counts all cycle products of $S_2$. Similarly, changing the notations $3^1$, $2^1 1^1$ and $1^1 1^1 1^1$ in Table \ref{table 2} by $g^1_3$, $g_2^1 g^1_1$ and $g_1 g_1 g_1 = g_1^3$ respectively, we see that $Y_3 = g_1^3 + 3 g_2 g_1 + g_3$ counts all cycle products of $S_3$. These observations illustrate the  notion of cycle index, which will now be defined.

Consider $G$ as the group whose elements $g$ are the permutations of $\mathbb{S}$, and let $Z_G (g_1,g_2, \ldots, g_n)$ be the polynomial in $n$ variables $g_1,g_2, \ldots, g_n$ such that for each $g \in G$, the type of $g$ is given by the product $z_g (g_1,g_2, \ldots, g_n)  = g_1^{m(1)} g_2^{(m(2)} \cdots  g_n^{m(n)}$ as the partition of $n$ of part $\lambda  = \left( 1^{m(1)} 2^{m(2)} \cdots n^{m(n)} \right)$, and such that $n = 1 \cdot m(1) + 2 \cdot m(2) + \ldots + n \cdot m(n)$. Then, the polynomial 

\begin{subequations}
\begin{align}
Z_G (g_1,g_2, \ldots, g_n) &= \frac{1}{|G|} \sum_{g \in G} z_g (g_1,g_2, \ldots, g_n)  \\ &= \frac{1}{|G|} \sum_f f \left( m(1), \ldots, m(n) \right)  \cdot g_1^{m(1)} g_2^{m(2)} \cdots  g_n^{m(n)} 
\end{align}
\end{subequations}

\noindent is defined as the cycle index of $G$, with $f \left( m(1), \ldots, m(n) \right)$ representing the number of permutations of type $\left( 1^{m(1)} 2^{m(2)} \cdots n^{m(n)} \right)$. The formula above is reminiscent of Burnside's lemma, except that here, one distinguishes the cycles of different lengths, and specifies the number of cycles there are. 

In the present case, we are interested in the cycle index of the symmetric group $S_n$, which is defined by the formula \cite{cameron_1994}

\begin{eqnarray}
\mathcal{Z} \left( S_n \right)= \sum_{c_1 + 2 c_2 + \ldots + l c_l=n} \frac{1}{\prod_{k=1}^n k^{c_k} c_k !} \prod_{k=1}^n x_k^{c_k}.
\end{eqnarray}

\noindent It is well known that the cycle index of the symmetric group $S_n$ can be expressed in terms of (complete) Bell polynomials as follows

\begin{eqnarray}
\label{cycle index bell polynomial}
Z(S_n) = \frac{Y_n (0!a_1, 1! a_2, \ldots, (n-1)! a_n)}{n!}.
\end{eqnarray}
Then, setting $a_n= \frac{1}{|1-q^n|^2}$, we can immediately identify the arguments in the above equation with the term $g_n= (n-1)! \frac{1}{|1-q^n|^2}$. $Z_{log} \left( q,\bar{q} \right)$ can therefore be rewritten as

\begin{eqnarray}
\label{z log is cycle index}
 Z_{\rm{log}} \left(\nu; q,\bar{q} \right) = \sum_{n=0} Z \left( S_n \right) \left( \nu \right)^n.
\end{eqnarray}

Eq. (\ref{z log is cycle index}) is the first indication that $Z_{\rm{log}} \left( \nu; q,\bar{q} \right)$ counts a collection of spaces under the action of the symmetric group. This preliminary result will be made more precise in the next part of this section from a representation theoretic point of view.


\subsection{$Z_{\rm{log}} \left( \nu; q,\bar{q} \right)$  as the generating function of 1-part Schur polynomials}
From a mathematics point of view, partitions are directly related to the representation theory of permutation groups. Indeed, the number of irreducible representations of a permutation group is equal to the number of orbits (or disjoint cycles) of the permutation group with respect to inner automorphism, or in other words to the number of conjugacy classes, which is equal to the number of partitions of the group order. As we have seen above, graphically this is represented by the number of Young diagrams. Physically, partitions play a important role in describing multi-particle systems. We make use of the theoretical results mentioned earlier to show that $Z_{\rm{log}} \left( \nu; q,\bar{q} \right)$ is the generating function of the 1-part Schur polynomial. This result allows for a more explicit description of the multi-particle system under investigation. 

Let $\mathbb{C}$ be the field of complex numbers, and $\mathcal{M}\mathcal{A}\mathcal{T}_n$ the set of all $n \times n$ matrices with entries in $\mathbb{C}$. The complex general linear group of degree $n$ denoted $GL_n$ is then the group of all $X= \left( x_{i,j}  \right)_{n \times n} \in \mathcal{M}\mathcal{A}\mathcal{T}_n$, and the group homomorphism $X: G \rightarrow GL_n$ is a matrix representation of a group $G$. Let $V$ be a vector space and $GL(V)$ the set of all invertible linear transformations of $V$ to itself. If dim $V=n$ (say for instance if $V$ is the vector space $\mathbb{C}^n$ given a matrix representation $X$ of degree $n$), then the group $GL(V)$ is isomorphic to $GL_n$, and it is possible to define the group homomorphism $\rho : G \rightarrow GL(V) $. In other words, the vector space $V$ carries a representation of $G$. As we will see now, the decomposition of tensor products of the representation $V$ plays a crucial role in the interpretation of the counting organised in $Z_{\rm{log}} \left( \nu; q,\bar{q} \right)$.

Let $V$ be a representation of $G$, and let $\rho$ be the associated group homomorphism $\rho : G \rightarrow GL(V) $. By tensoring $n$ copies of $V$, one gets the space $V^{\otimes n}$. If $v_1, \ldots, v_d$ is a basis of $V$, then a basis of $V^{\otimes n}$ is the collection of vectors $v_{k_{1}}, \ldots, v_{k_{n}}$, where the indexes $k_1, \ldots, k_n$ range over $\left\{ 1 \cdots d \right\}^n$, so that $V^{\otimes n}$ has dimension $d^n$. We take $G$ to be the symmetric group $S_n$. $V^{\otimes n}$ admits an action of $S_n$ by permutation of the factors of $V$ with $V^{\otimes n}$. In other words, re-expressing the vector basis of $V^{\otimes n}$ in braket notation as $\ket{i_1 i_2 \cdots i_n}$, the permutation $\pi \in S_n$ is the map $\pi: V^{\otimes n} \rightarrow V^{\otimes n}$ which acts on the basis as 

\begin{eqnarray}
\pi \ket{i_1 i_2 \cdots i_n} \rightarrow \ket{i_{\pi(1)} i_{\pi(2)} \cdots i_{\pi(n)}}.
\end{eqnarray}

\noindent The matrix elements of $\pi$ have the form

\begin{eqnarray}
\pi^{I}_{J} = \bra{i_1 i_2 \cdots i_n} \pi \ket{i_{\pi(1)} i_{\pi(2)} \cdots i_{\pi(n)}},
\end{eqnarray}

\noindent where $I$ and $J$ stand for $i_1 i_2 \cdots i_n$ and $j_1 j_2 \cdots j_n$ respectively. 

The action of $S_n$ on $V^{\otimes n}$ is reducible, however we can introduce the projection operators 

\begin{eqnarray}
\mathcal{P}_{R} = \frac{1}{n!} \sum_{\pi \in S_n} \chi_R(\pi) \pi  \hspace{0.5cm}\Longleftrightarrow \hspace{0.5cm} \left( \mathcal{P}_{R}  \right)^I_J =  \frac{1}{n!} \sum_{\pi \in S_n} \chi_R(\pi) \left( \pi \right)^I_J,
\end{eqnarray}

\noindent where $R$ is a Young diagram with $n$ rows and $\chi_R(\pi)$ is the character of the matrix representing $\pi$ in the irreducible representation $R$, such that $\left( \mathcal{P}_{R}  \right)^I_J$ act on $V^{\otimes n}$ projecting onto the irreducible representations contained in $V^{\otimes n}$. 

Let us now consider the complex matrix $X \in GL(V)$. By tensoring $n$ copies of $X$, one gets an operator $X^{\otimes n}$ that acts on the space $V^{\otimes n}$. Denoting the matrix elements by 

\begin{eqnarray}
\left( X^{\otimes n} \right)^I_J = X^{i_1}_{j_1} X^{i_2}_{j_2} \cdots X^{i_n}_{j_n},
\end{eqnarray}

\noindent it is possible to write 

\begin{eqnarray}
\label{schur polynomials definition}
\chi_R(X) = \left( \mathcal{P}_{R}  \right)^I_J \left( X^{\otimes n} \right)^J_I = tr \left( \mathcal{P}_{R} X^{\otimes n} \right) = \frac{1}{n!} \sum_{\pi \in S_n} \chi_R(\pi) \cdot tr \left( \pi X^{\otimes n} \right), 
\end{eqnarray}

\noindent where Eq. (\ref{schur polynomials definition}) defines the Schur polynomials $\chi_R(X)$. Schur polynomials are very effective in the description of multi-trace structures. Indeed, $tr \left( \pi X^{\otimes n} \right)$ is a single trace structure in the space $V^{\otimes n}$. However, specifying $\pi$ and contracting the indices in the matrix elements $\left( X^{\otimes n} \right)^I_J$ produce a multi-trace structure. Appendix \ref{Properties of Schur polynomials} shows the properties of Schur polynomials and Young diagrams for $S_2$ and $S_3$. In general, any multi-trace structure involving $n$ $X$ matrices may be obtained from a single trace of an $S_n$ permutation acting on $X^{\otimes n}$ in $V^{\otimes n}$. In particular, by writing

\begin{eqnarray}
\left( tr X^n \right) = \frac{1}{\left( 1 - q^n \right) \left( 1- \bar{q}^n \right)},
\end{eqnarray}

\noindent it is possible to identify the first three 1-part Schur polynomials as

\begin{subequations}
\begin{align}
& \chi_{\Yboxdim{4pt} \yng(1)} = (tr X), \\
& \chi_{\Yboxdim{4pt} \yng(2)} = \frac{1}{2!} \left( \left( tr  X \right)^2 + \left( tr  X^2 \right) \right), \\
& \chi_{\Yboxdim{4pt} \yng(3)} = \frac{1}{3!} \left( \left( tr  X \right)^3 + 3 \left( tr  X \right) \left( tr X^2 \right) + 2 \left( tr  X^3 \right) \right),
\end{align}
\end{subequations}

\noindent At last, it appears that $Z_{log} \left( \nu; q, \bar{q}  \right)$ is the generating function of 1-part Schur polynomials expressed as 

\begin{eqnarray}
\label{Z as schur polynomials}
Z_{log} \left( \nu; q, \bar{q}  \right) = \sum_{n=0}^{\infty} \chi_{\underbrace{\Yboxdim{4pt} \yng(1) \cdots \yng(1)}_n}  \cdot (\nu)^n.
\end{eqnarray}

\subsection{Interpretation}
We close this section by giving a symmetric group interpretation of the results obtained in \cite{Mvondo-She:2018htn}, and discussed above. A large part of the combinatorial work interpreted from a representation theoretic perspective describes the implementation of a bosonic statistics to obtain the multiparticle contribution from the single particle one. The space of single particle is being tensored $n$-times and only the symmetric part of the product is being retained. The projectors $\mathcal{P}_R$ onto the symmetric part give rise to Schur polynomials labelled by Young diagrams with a single row as apparent in Eq. (\ref{Z as schur polynomials}), indicating that $Z_{log} \left( \nu; q, \bar{q}  \right)$ is the generating function of the (single row) 1-part Schur polynomials. In the next section, we discuss this interpretation using a well known theorem of invariant theory, that will allow us to give a novel description of the moduli space of the logarithmic states in critical massive gravities.


\section{Moduli space of logarithmic states}
In section \ref{representation theory section}, we made use of the expression of $Z_{\rm{log}} \left( \nu; q,\bar{q} \right)$ in terms of Bell polynomials to show how it encodes information about spaces invariant under the action of the symmetric group. From a different perspective, the Bell polynomial formulation will also be useful (section \ref{HW algebra subsection}) in the combinatorial description of a Fock space created by the action of generators of the Heisenberg-Weyl algebra. In this section, we use the fact that $Z_{\rm{log}} \left( \nu; q,\bar{q} \right)$ can also be expressed as the generating function of \textit{Hilbert series} to show the symmetric product orbifold structure of the moduli space of the logarithmic states.

Also called \textit{Molien} or \textit{Poincar\'{e}} function, the Hilbert series is a generating function familiar in algebraic geometry for counting the dimension of graded components of the coordinate ring (see Appendix \ref{Invariant theory}). Its approach has been developed and extensively used in theoretical physics with for instance the work of \cite{Pouliot:1998yv}, and notably with several applications under the so-called Plethystic Program initiated in \cite{Benvenuti:2006qr,Feng:2007ur}. In connection with the Plethystic Program, the Hilbert series have been the essential instrument of a systematic method that yields the generating function of multi-trace operators in gauge theory from the generating function of single-trace operators at large N. This formalism was shown to hold in the present setting of a function generating multi particle states from single particle ones. Indeed, we recall from \cite{Mvondo-She:2018htn} that the multivariate Hilbert series 

\begin{eqnarray}
\label{Hilbert series}
\mathcal{G}_{1} (q,\bar{q}) =  \sum^{\infty}_{m \geq 0,\bar{m} \geq 0}  q^{m}  \bar{q}^{\bar{m}} = \frac{1}{|1-q|^2}= \frac{1}{(1-q) (1-\bar{q})},
\end{eqnarray}
that counts the single particle contribution, can be acted upon by the bosonic plethystic exponential $PE^{\mathcal{B}}$ to generate new partition functions such that  

\begin{eqnarray}
\label{PE}
Z_{\rm{log}} \left( \nu; q,\bar{q} \right) = PE^{\mathcal{B}} \left[ \mathcal{G}_{1} (q,\bar{q}) \right] =  \exp \left( \sum^{\infty}_{n=1} \frac{\left( \nu \right)^n}{n}  \mathcal{G}_1 \left( q^n,\bar{q}^n \right) \right),
\end{eqnarray}
with 

\begin{eqnarray}
\mathcal{G}_1 \left( q^n,\bar{q}^n \right) = \frac{1}{|1-q^n|^2} = \frac{1}{(1-q^n) (1-\bar{q}^n)}. 
\end{eqnarray} 
The connection between the plethystic exponential in Eq. (\ref{PE}) and the cycle index discussion of section \ref{representation theory section} in terms of Bell polynomials can be made clear by noting that $a_1$ in Eq. (\ref{cycle index bell polynomial}) can be identified with the function $\mathcal{G}_1 \left( q, \bar{q} \right)$, and accordingly,  $a_n \equiv a_n \left( q, \bar{q} \right) = \mathcal{G}_1 \left( q^n,\bar{q}^n \right)$. In analogy with the aforementioned applications, this shows that the Hilbert series $\mathcal{G}_1 {\left( q, \Bar{q} \right)}$ counts single particle states, while the plethystic exponential $PE^{\mathcal{B}}\left[ \mathcal{G}_1 {\left( q, \Bar{q} \right)} \right]$ counts multi-particle states.

The formalism of Hilbert series acted upon by the plethystic exponential is well known for its use in describing algebraic and geometric aspects of moduli space. We will draw from that knowledge to study the configuration space of logarithmic states. Essential to this will be a discussion of symmetric products in the spirit of \textit{Hilbert schemes of points on surfaces} \cite{nakajimalectures}. 

\subsection{Hilbert schemes of points on surfaces and symmetric products}
The Hilbert scheme $X^{[n]}$ of points on a surface is a simple example of a moduli space. It consists in the description of the configuration space of $n$ points on $X$, $i.e$ the space of unordered $n$-tuples of points of $X$ \cite{nakajimalectures}. 

Formally, the Hilbert scheme of points can be defined as

\begin{eqnarray}
\label{hilbert scheme from ideals}
X^{[n]} := \left\{  \mbox{{\it{I}} | {\it{I}} is an ideal of $X[x_1, \ldots,x_n]$ with dim($X[x_1, \ldots,x_n]/I$)={\it{n}}}  \right\},
\end{eqnarray}

\noindent where $X[x_1, \ldots,x_n]$ is the coordinate ring of $X$. In the above definition, $X^{[n]}$ is considered as a set. It can be defined in a more geometric flavor as

\begin{eqnarray}
X^{[n]} := \left\{  \mbox{ {$Q_Z$} | {\it{$Q_Z$}} is a quotient ring of {\it{X}} with dim({\it{$Q_Z$}})={\it{n}} }  \right\}.
\end{eqnarray}

\noindent The algebra-geometry correspondence is expressed as

\begin{eqnarray}
0 \rightarrow I \rightarrow X \rightarrow Q_Z \rightarrow 0,
\end{eqnarray}

\noindent where $Z$ is the 0-dimensional subscheme of $X$, and $Q_Z$ is the coordinate ring of $Z$, and allows for flexibility of terminology between schemes and ideals.

The construction of a moduli space such as the Hilbert scheme can be accomplished by taking the quotient of $X$ by a group acting on it. In that endeavor, it is sensible to consider the quotient by the action of the symmetric group, since we do not distinguish between points. This gives the symmetric product 

\begin{eqnarray}
S^n X = \underbrace{X \times X \times \cdots \times X}_{n \hspace{0.1cm} \mbox{times}}/S_n.
\end{eqnarray}

\noindent However, the symmetric product $S^n X$ (also denoted $X^{(n)}$) is singular. Indeed, if for instance we consider the case $n=2$, the group action is not free along the diagonal $\mathcal{D} \subset X \times X$, which yields a singular locus along the diagonal in $S^2 X$. More precisely, approaching the diagonal corresponds in $X^2$ to the two points approaching each other, and eventually overlapping. At that stage, the system has lost one degree of freedom. A possible resolution of the problem would be to keep track of the direction the two points approach each other along. That is in fact the difference between $X^{[n]}$ and $S^n X$: The Hilbert scheme $X^{[n]}$ is a resolution of singularities of the symmetric product $S^n X$. When there exist $n$ distinct points $p_1, \ldots, p_n$ in $X$, each point defines both a point in $X^{[n]}$ and a point in $S^n X$ \cite{nakajimaquiver}, and setting the ideal of Eq. (\ref{hilbert scheme from ideals}) to

\begin{eqnarray}
I := \left\{ \mbox{$f \in X[x_1, \ldots,x_n] \hspace{0.2cm} | \hspace{0.2cm} f(p_1) =  \cdots = f(p_n)=0$}  \right\},
\end{eqnarray}
$I$ is indeed an ideal with dim($X[x_1, \ldots,x_n]/I$)={\it{n}}. This is the case when $\rm{dim} X=1$ ($i.e$ $n=1$): the Hilbert scheme $X^{[n]}$ is isomorphic to the $n$-th symmetric product $S^n X$ and we have

\begin{eqnarray}
X^{[n]} \simeq  S^n X.
\end{eqnarray}

\noindent A different situation is when some points collide. Looking at the case $n=2$, two types of ideals must be taken into account in $X^{[2]}$. One can either consider an ideal given by two distinct points $p_1$ and $p_2$, or the ideal

\begin{eqnarray}
I= \left\{ f \hspace{0.2cm} | \hspace{0.2cm} f(p)= 0, \hspace{0.2cm} df_P(v) = 0 \right\},
\end{eqnarray}
where $p$ is a point of $X$ and $v$ is a vector in the tangent space $T_p X$. The information of the direction in which $p_1$ approaches $p_2$ is remembered in this ideal. In the symmetric product, this information is lost and one just has $2p$. When $n >2$, more complicated ideals appear.

\subsection{The $n$-th symmetric product of $\mathbb{C}^2$}
In the spirit of the Hilbert scheme of points on surfaces briefly discussed above, we consider the case when $X= \mathbb{C}^2$. More precisely, we consider the family 

\begin{eqnarray}
S^n (\mathbb{C}^2) \simeq \mathbb{C} \left[ x_1,y_1; x_2,y_2; \ldots ; x_n,y_n \right] /S^n,
\end{eqnarray}
where $(x,y)$ are the coordinates of $\mathbb{C}^2$ and $S_n$ permutes the $n$-tuple of variables $(x_i,y_i)$. The computation of the Hilbert series of the invariant ring $S^n (\mathbb{C}^2)$ then amounts to extending Molien's Theorem to the bi-graded case. Such extension has already been studied \cite{forger1998invariant,traves2006differential}. We now show that in our case, $Z_{log} (\nu;q,\bar{q})$ is the generating function of a $\nu$-inserted bi-graded Molien series of the symmetric group and can be expressed as 

\begin{eqnarray}
Z_{log} (\nu;q,\bar{q}) \equiv Z_{log} (\nu,q,\bar{q}; \mathbb{C}^2)= \sum_{n=0} Z_n \left( q,\bar{q}; \mathbb{C}^2 \right) \left( \nu \right)^n.
\end{eqnarray}

\noindent We start by looking at the cases when $S_2$ and $S_3$ act on the standard basis of $\mathbb{C}^2$ and $\mathbb{C}^3$, respectively.

The symmetric group on two objects can be presented as $S_2 = < e, \sigma >$, where $e$ is the identity element and $\sigma$ can be expressed in cycle notation as $\sigma = (12)$. If we consider the action of $S_2$ on $2 \times 2$ matrices by permuting coordinates $q$, we have

\begin{eqnarray}
\begin{pmatrix} q & 0 \\ 0 & q \end{pmatrix} \xrightarrow[]{\pi (e)} \begin{pmatrix} q & 0 \\ 0 & q \end{pmatrix} \hspace{0.5cm} \mbox{and} \hspace{0.5cm} \begin{pmatrix} q & 0 \\ 0 & q \end{pmatrix} \xrightarrow[]{\pi (\sigma)} \begin{pmatrix} 0 & q \\ q & 0 \end{pmatrix} .
\end{eqnarray}

\noindent Then

\begin{subequations}
\begin{align}
\det \left( I - q \pi (e) \right) &= \det \left[ \begin{pmatrix} 1 & 0 \\ 0 & 1  \end{pmatrix} - \begin{pmatrix} q& 0 \\ 0 & q  \end{pmatrix} \right] \\ &= \det
\begin{pmatrix} 1-q & 0 \\ 0 & 1-q  \end{pmatrix} \\ &= \left( 1-q \right)^2, 
\end{align}
\end{subequations}
 
\noindent and

\begin{subequations}
\begin{align}
\det \left( I - q \pi (\sigma) \right) &= \det \left[ \begin{pmatrix} 1 & 0 \\ 0 & 1  \end{pmatrix} - \begin{pmatrix} 0 & q \\   q & 0 \end{pmatrix} \right] \\ &= \det
\begin{pmatrix} 1 & -q \\ -q & 1  \end{pmatrix} \\ &= \left( 1-q^2 \right). 
\end{align}
\end{subequations}

\noindent In the same way, the action of $S_2$ on $2 \times 2$ matrices by the permuting coordinates $\bar{q}$ allows us to write $\det \left( I - \bar{q} \pi (e) \right)=\left( 1-\bar{q} \right)^2$ and $\det \left( I - \bar{q} \pi (\sigma) \right)=\left( 1-\bar{q}^2 \right)$. From there, we have

\begin{subequations}
\begin{align}
Z_2\left( q, \bar{q}; \mathbb{C}^2 \right)  &= \frac{1}{2!} \left[ \mathcal{G}_1 \left( q, \bar{q} \right)^2 + \mathcal{G}_1 \left( q^2, \bar{q}^2 \right) \right] \\ &= \frac{1}{2!} \left[ \frac{1}{(1-q)^2(1-\bar{q})^2} + \frac{1}{(1-q^2)(1-\bar{q}^2)} \right] \\ &= \frac{1}{2!} \left[ \frac{1}{\det \left( I - q e \right) \det \left( I - \bar{q} e \right)} + \frac{1}{\det \left( I - q \sigma \right)  \det \left( I - \bar{q} \sigma \right)} \right].
\end{align}
\end{subequations}

Next we consider the symmetric group $S_3 = < e, \sigma, \tau >$, where the elements correspond respectively to the identity, the three-element conjugacy class that consists of swaping two coordinates, and the two-element conjugacy class of cyclic permutations. Using cycle notations, $\sigma = \left\{ (12),(13),(23) \right\}$ and $\tau = \left\{ (123),(132) \right\}$.
The action of $S_3$ on $3 \times 3$ matrices by permuting coordinates $q$ would yield the following. Starting from the identity

\begin{eqnarray}
\det
\begin{pmatrix}
1-q & 0 & 0 \\ 0 & 1-q & 0 \\ 0 & 0 & 1-q
\end{pmatrix}
= ( 1-q)^3.
\end{eqnarray}

\noindent Then, taking one of the three terms consisting of swaps of two coordinates, say $\sigma = (12)$ identically results into 

\begin{eqnarray}
\det
\begin{pmatrix}
1 & -q & 0 \\ -q & 1 & 0 \\ 0 & 0 & 1-q
\end{pmatrix}
= (1-q)( 1-q^2).
\end{eqnarray}

\noindent Finally, taking one of the two cyclic permutation terms, say $\tau = (132)$, identically yields

\begin{eqnarray}
\det
\begin{pmatrix}
1 & -q & 0 \\ 0 & 1 & -q \\ -q & 0 & 1
\end{pmatrix}
= (1-q^3).
\end{eqnarray}

\noindent Acting in the same way on coordinates $\bar{q}$ allows us to eventually write

\begin{subequations}
\begin{align}
& Z_3\left( q, \bar{q}; \mathbb{C}^2 \right) = \frac{1}{3!} \left[ \mathcal{G}_1 \left( q, \bar{q} \right)^3 + 3 \mathcal{G}_1 \left( q, \bar{q} \right) \mathcal{G}_1 \left( q^2, \bar{q}^2 \right) + 2 \mathcal{G}_1 \left( q^3, \bar{q}^3 \right) \right]\\ 
&= \frac{1}{3!} \left[  \frac{1}{(1-q)^3 (1-\bar{q})^3}  + \frac{3}{(1-q)(1-q^2)(1-\bar{q})(1-\bar{q}^2)} + \frac{2}{(1-q^3)(1-\bar{q}^3)}  \right] \\ &= \frac{1}{3!} \left[ \frac{1}{\det \left( I - q \pi(e) \right) \det \left( I - \bar{q} \pi (e) \right)} + \frac{3}{\det \left( I - q \pi (\sigma) \right) \det \left( I - \bar{q} \pi (\sigma) \right)} + \frac{2}{\det \left( I - q \pi (\tau) \right) \det \left( I - \bar{q} \pi (\tau) \right)}\right].
\end{align}
\end{subequations}

\noindent We can generalize the above procedure for all permutations $\pi$ of the elements $ g \in S_n$ as follows. We recall that the \textit{cycle type} of a permutation $\pi$ is the integer vector $m(\pi) = \left( m(1),m(2), \ldots, m(n) \right)$, where $m(k)$ are the multiplicities that count the number of cycles of length $k$ in the cycle decomposition of $\pi$. Then, applying the above discussion to $S_n$ acting on the space $\left( \mathbb{C} \times \mathbb{C} \right)^n$ by permuting the $q$- and $\bar{q}$-coordinates, if $\pi$ is the permutation of $S_n$ with cycle type $m(\pi) = \left( m(1),m(2), \ldots, m(n) \right)$, the standard bases of $\left( \mathbb{C} \times \mathbb{C} \right)^n$ decomposes into cycles of length $k$ such that

\begin{eqnarray}
\det \left(I_n - q \pi \right) \det \left(I_n - \bar{q} \pi \right) = \prod_{k=1}^n \left[ \left( 1-q^k \right) \left( 1- \bar{q}^k \right) \right]^{m(k)}.
\end{eqnarray}

\noindent As a result

\begin{eqnarray}
\frac{1}{\det \left(I_n - q \pi \right) \det \left(I_n - \bar{q} \pi \right)}= \prod_{k=1}^n \left[ \mathcal{G}_1 \left( q^k, \bar{q}^k \right) \right]^{m(k)}, 
\end{eqnarray}

\noindent and finally

\begin{eqnarray}
\label{Z as bigraded molien series}
Z_{log} (\nu;q,\bar{q}) \equiv Z_{log} (\nu,q,\bar{q}; \mathbb{C}^2) = \sum_{n=0}^{\infty} \frac{1}{\left| S_n \right|} \sum_{g \in S_n} \frac{(\nu)^n}{\det \left( I - q \pi (g)\right)  \det \left( I - \bar{q} \pi (g)\right)}.
\end{eqnarray}

\noindent  $Z_{log} (\nu,q,\bar{q}; \mathbb{C}^2)$ is therefore the generating function of a ($\nu$-inserted) bi-graded Molien series of $S_n$, and of Hilbert series of the ring of invariants $S^{n}(\mathbb{C}^2)$. 

Closer to our previous discussion on Hilbert schemes of points on surfaces, generating functions taking the exponential form of $Z_{log} (\nu,q,\bar{q}; \mathbb{C}^2)$ were considered in \cite{Nakajima:2003pg} and in \cite{Feng:2007ur}. Using our notation, we can then write 

\begin{eqnarray}
Z_{log} (\nu,q,\bar{q}; \mathbb{C}^2) = PE^{\mathcal{B}} \left[ \frac{1}{(1-q)(1-\bar{q})} \right] =  \exp \left( \sum^{\infty}_{n=1} \frac{\left( \nu \right)^n}{n (1-q^n)(1-\bar{q}^n)}   \right).
\end{eqnarray}

\noindent In the simplest case of CCTMG for instance, $i.e$ when $\nu= q^2$, one can write 

\begin{eqnarray}
\label{Z log for CCTMG}
Z_{log} (q^2;q,\bar{q}; \mathbb{C}^2) = PE^{\mathcal{B}} \left[ \frac{1}{(1-q)(1-\bar{q})} \right] =  \exp \left( \sum^{\infty}_{n=1} \frac{\left( q^2 \right)^n}{n (1-q^n)(1-\bar{q}^n)}   \right).
\end{eqnarray}

\noindent $S^n (\mathbb{C}^2)$ is an orbifold locally isomorphic to an open set of the Euclidean space quotiented by the action of the symmetric group \cite{nakajimaquiver,Nakajima:2003pg}. The above analysis therefore brings forth the orbifold structure of the moduli space of logarithmic partners states in critical massive gravities. 


\subsection{Differential operators on orbifolds}
\label{HW algebra subsection}
In this section, using an invariant theoretic language, we revisit some of the work done in \cite{Mvondo-She:2018htn}, and give an interpretation to the hidden structures found in the study of $Z_{log} (\nu;q,\bar{q}; \mathbb{C}^2)$.

Fock spaces were designed as an algebraic framework to construct many-particle states in quantum mechanics. They typically represent the state space of an indefinite number of identical particles (an electron gas, photons, etc...). These particles can be classified in two types, bosons and fermions, and their Fock spaces look quite different. The reason why a Fock space are of great interest is that several important algebras can naturally act on it. Fermionic Fock spaces are naturally representations of a Clifford algebra, where the generators correspond to adding or removing a particle in a given energy state. In a similar way, bosonic Fock space is naturally a representation of a Weyl algebra. 

Returning to the discussion on Hilbert schemes, for the non-compact space $\mathbb{C}^2$, a connection between the theory of Hilbert schemes of points on surfaces and the infinite dimensional Heisenberg algebra was made through the construction of a representation of the Heisenberg algebra on the homology group of the Hilbert scheme, turning the homology group into a Fock space \cite{nakajima1997heisenberg}. The construction showed that the Fock space representation on the polynomial ring of infinitely many variables is an important representation of the Heisenberg algebra. In the present case, with an interest on the bosonic Fock space, we discuss the construction of a combinatorial model of creation and annihilation operators that are generators of a Heisenberg-Weyl algebra and that act on the bosonic Fock space. 

In general, a Fock space is considered on a Hilbert space, but in the simplest case and for the purpose of our discussion, the bosonic vector space is obtained by considering a complex vector space $\mathbb{C}$. Then, the bosonic Fock space as a vector space is essentially a space of polynomials of infinitely many variables. A typical basis can be constructed using Schur symmetric functions. In our case, we consider a space of Bell polynomials $Y_n$ of infinitely many variables $ g_1, g_2, \cdots, g_n$. 

Given the $n$-dimensional polynomial ring $\mathbb{C}[x_1, \ldots, x_n]$, a subring of invariant polynomials denoted $\mathbb{C}[x_1, \ldots, x_n]^{S_n}$ can be constructed from invariants $g_1, \ldots, g_n \in \mathbb{C}[x_1, \ldots, x_n]$. The polynomials of the invariant ring $\mathbb{C}[x_1, \ldots, x_n]^{S_n}$ can take the form of Bell polynomials $Y_n$ with coordinates $g_1, \ldots, g_n$  such that $\mathbb{C}[x_1, \ldots, x_n]^{S_n} \simeq\mathbb{C}[g_1, \ldots, g_n]$. From the ring of differential operators

\begin{eqnarray}
D\left( \mathbb{C}[g_1, \ldots, g_n] \right) = \mathbb{C}<g_1, \ldots, g_n, \partial_{g_1}, \ldots, \partial_{g_n}>,
\end{eqnarray}
we construct the multiplication operator $\hat{X}=g_1 + \sum^{\infty}_{k=1} g_{k+1} \frac{\partial}{\partial g_k}$ and the derivative operator $\hat{D}=\frac{\partial}{\partial g_1}$ such that they satisfy the Heisenberg-Weyl algebra $\left[ \hat{X}, \hat{D} \right]=1$. Then, defining the Hilbert series 

\begin{eqnarray}
 Z \left( \mathbb{C}[x_1, \ldots, x_n]^{S_n} \right) \equiv  Z\left( S^{n}(\mathbb{C}) \right) = \frac{Y_n(g_1, \ldots, g_n)}{n!},
\end{eqnarray}

\noindent these operators act as ladder operators on the Hilbert series $Z \left( \mathbb{C}[x_1, \ldots, x_n]^{S_n} \right)$ at each $n$ level in the following way.

\begin{proposition}
\label{prop}
Let  $Z\left( S^{n}(\mathbb{C}) \right)$ be defined in terms of Bell polynomials $Y_n$ as

\begin{eqnarray}
Z\left( S^{n}(\mathbb{C}) \right) = \frac{Y_n(g_1, \ldots, g_n)}{n!}.
\end{eqnarray}

\noindent The set of operators

\begin{eqnarray}
\hat{X}=g_1 + \sum^{\infty}_{k=1} g_{k+1} \frac{\partial}{\partial g_k}; \hspace{1cm}
\hat{D}=\frac{\partial}{\partial g_1},\nonumber 
\end{eqnarray}
generating the Heisenberg-Weyl algebra $[\hat{D},\hat{X}]=1$, $[\hat{D},1]=[\hat{X},1]=0$ acts on $Z\left( S^{n}(\mathbb{C}) \right)$ as

\begin{subequations}
\begin{align}
\hat{X} Z\left( S^{n}(\mathbb{C}) \right) &= (n+1) Z\left( S^{n+1}(\mathbb{C}) \right)  \\ \hat{D} Z\left( S^{n}(\mathbb{C}) \right) &= Z\left( S^{n-1}(\mathbb{C}) \right)  \\
\hat{X} \hat{D} Z\left( S^{n}(\mathbb{C}) \right) &= n Z\left( S^{n}(\mathbb{C}) \right).
\end{align}
\end{subequations}
\end{proposition}

\begin{proof}
From \cite{Mvondo-She:2018htn}, it is known that

\begin{subequations}
\begin{align}
\hat{X} Y_n &= Y_{n+1},  \\
\hat{D} Y_n &= n Y_{n-1},  \\
\hat{X} \hat{D} Y_n &= n Y_n. 
\end{align}
\end{subequations}

\noindent Hence we can write

\begin{subequations}
\begin{align}
\hat{X} \left[ n! Z\left( S^{n}(\mathbb{C}) \right) \right] &= (n+1)! Z\left( S^{n+1}(\mathbb{C}) \right)  \\  \Rightarrow \hspace{0.5cm} \hat{X} \left( Z\left( S^{n}(\mathbb{C}) \right) \right) &= (n+1) Z\left( S^{n+1}(\mathbb{C}) \right). 
\end{align}
\end{subequations}

\noindent Similarly

\begin{subequations}
\begin{align}
\hat{D} \left[ n! Z\left( S^{n}(\mathbb{C}) \right) \right] &= n \left[ (n-1)!  Z\left( S^{n-1}(\mathbb{C}) \right) \right]  \\ \Rightarrow \hspace{0.5cm}
\hat{D} \left( Z\left( S^{n}(\mathbb{C}) \right) \right) &= Z\left( S^{n-1}(\mathbb{C}) \right),  
\end{align}
\end{subequations}

\noindent and

\begin{subequations}
\begin{align}
\hat{X} \hat{D} \left[  Z\left( S^{n}(\mathbb{C}) \right) \right] &= n \left[  Z\left( S^{n}(\mathbb{C}) \right) \right]  \\ \Rightarrow \hspace{0.5cm} \hat{X} \hat{D} Z\left( S^{n}(\mathbb{C}) \right) &= n Z\left( S^{n}(\mathbb{C}) \right).
\end{align}
\end{subequations}
\end{proof}

\section{Moduli space of logarithmic states and Calabi-Yau orbifolds}
In the previous section, we used insights from the plethystic program and Molien's theorem to show that the moduli space of the logarithmic states is given by the symmetric product $S^{n}(\mathbb{C}^2)$. In this section, we take a further geometric perspective, by establishing that the moduli space is a Calabi-Yau space. 

The context of the moduli space of the logarithmic states is one in which as seen earlier, the Bell polynomials with coordinates $g_1, g_2, \ldots, g_n$ form a ring on the space $S^{n}(\mathbb{C}^2)$. The Hilbert series of the polynomial ring has a very interesting palindromic property, which can be stated as follows.

\begin{proposition}
Let the generating function 

\begin{eqnarray}
Z_{log}\left( \nu;q,\bar{q};\mathbb{C}^2 \right) = \sum_{n=0}^{\infty} Z_n \left( q,\bar{q};\mathbb{C}^2 \right) (\nu)^n.
\end{eqnarray}

\noindent The numerator of the Hilbert series $Z_n \left( q,\bar{q};\mathbb{C}^2 \right)$ is palindromic, i.e it can be written in the form of a degree $m,\bar{m}$ polynomial in $q,\bar{q}$

\begin{eqnarray}
P_{m,\bar{m}}( q, \bar{q})= \sum_{k=0}^{m} \sum_{\bar{k}=0}^{\bar{m}} a_{k,\bar{k}} q^k \bar{q}^{\bar{k}} ,
\end{eqnarray}

\noindent with symmetric coefficients $a_{m-k,\bar{m}-\bar{k}}=a_{k,\bar{k}}$, and $P_{m,\bar{m}}( 1, 1) \neq 0$.
\end{proposition}

\begin{proof}
We make use of a theorem by Stanley \cite{stanley1978hilbert} to show that if the numerator of the Hilbert series $Z_n \left( q,\bar{q};\mathbb{C}^2 \right)$ is palindromic, then the Hilbert series $Z_n \left( q,\bar{q};\mathbb{C}^2 \right)$ enjoys the following transformation property

\begin{eqnarray}
Z_n \left( \frac{1}{q}, \frac{1}{\bar{q}}  \right) = \left(q \bar{q} \right)^n Z_n \left( q,\bar{q} \right).
\end{eqnarray}

\noindent In order to show that the above equation holds, given that $Z_n \left( q,\bar{q};\mathbb{C}^2 \right)= \frac{Y_n}{n!}$ it suffices to show that by introducing the variable 

\begin{eqnarray}
\tilde{g}_n= g \left( \frac{1}{q^n}, \frac{1}{\bar{q}^n} \right),
\end{eqnarray}

\noindent then 

\begin{eqnarray}
\label{transformation of Bell polynomial}
Y_n \left( \tilde{g}_1, \tilde{g}_2, \ldots, \tilde{g}_n   \right) = \left(q \bar{q} \right)^n Y_n \left( g_1, g_2, \ldots, g_n \right).  
\end{eqnarray}

\noindent We start by writing $\tilde{g}_n$ in terms of $g_n$

\begin{eqnarray}
\label{transformation of g}
\tilde{g}_n = g \left( \frac{1}{q^n}, \frac{1}{\bar{q}^n} \right) = \frac{1}{1-\frac{1}{q^n}} \frac{1}{1- \frac{1}{\bar{q}^n}} = \left( q \bar{q} \right)^n g (q^n,\bar{q}^n) = \left( q \bar{q} \right)^n g_n.
\end{eqnarray}

\noindent Next, we use the determinantal form of the Bell polynomial \cite{johnson2002curious} 

\begin{eqnarray}
Y_n(g_1, \ldots, g_n) =
\det \begin{pmatrix}
\begin{pmatrix} n-1 \\ 0 \end{pmatrix} g_1 & \begin{pmatrix} n-1 \\ 1 \end{pmatrix} g_2 & \begin{pmatrix} n-1 \\ 2\end{pmatrix} g_3 & \cdots & \begin{pmatrix} n-1 \\ n-2 \end{pmatrix} g_{n-1} & \begin{pmatrix} n-1 \\ n-1 \end{pmatrix} g_n \\
-1 & \begin{pmatrix} n-2 \\ 0 \end{pmatrix} g_1 & \begin{pmatrix} n-2 \\ 1 \end{pmatrix} g_2 & \cdots & \begin{pmatrix} n-2 \\ n-3 \end{pmatrix} g_{n-2} & \begin{pmatrix} n-2 \\ n-2 \end{pmatrix} g_{n-1} \\ 
0 & -1 & \begin{pmatrix} n-3 \\ 0 \end{pmatrix} g_1 & \cdots & \begin{pmatrix} n-3 \\ n-4  \end{pmatrix} g_{n-3} & \begin{pmatrix} n-3 \\ n-3 \end{pmatrix} g_{n-2} \\
\vdots & \vdots & \vdots & & \vdots & \vdots \\
0 & 0 & 0 & \cdots & \begin{pmatrix} 1 \\ 0 \end{pmatrix} g_1 & \begin{pmatrix} 1 \\ 1 \end{pmatrix} g_2 \\
0 & 0 & 0 & \cdots & -1 & \begin{pmatrix} 0 \\ 0 \end{pmatrix} g_1
\end{pmatrix},
\end{eqnarray}

\noindent and prove Eq. (\ref{transformation of Bell polynomial}) by induction. \\ \\
For $n=1$, using Eq. (\ref{transformation of g}) we write

\begin{eqnarray}
Y_1 \left( \tilde{g}_1 \right) = \tilde{g}_1 =  \left( q \bar{q} \right)^1 g_1 = \left( q \bar{q} \right)^1 Y_1, 
\end{eqnarray}

\noindent showing that Eq. (\ref{transformation of Bell polynomial}) holds at level one. \\ \\ \noindent We assume that the identity holds at any level $n$, and show that if it holds at level $n$, it also holds at level $n+1$. \\ \\ \noindent
The determinantal form of the Bell polynomials at level $n+1$ can be expressed as \cite{johnson2002curious}

\begin{eqnarray}
Y_{n+1}(g_1, \ldots, g_{n+1}) =
\det \begin{pmatrix}
\begin{pmatrix} n \\ 0 \end{pmatrix} g_1 & \begin{pmatrix} n \\ 1 \end{pmatrix} g_2 & \begin{pmatrix} n \\ 2\end{pmatrix} g_3 & \cdots & \begin{pmatrix} n \\ n-1 \end{pmatrix} g_{n} & \begin{pmatrix} n \\ n \end{pmatrix} g_{n+1} \\
-1 & \begin{pmatrix} n-1 \\ 0 \end{pmatrix} g_1 & \begin{pmatrix} n-1 \\ 1 \end{pmatrix} g_2 & \cdots & \begin{pmatrix} n-1 \\ n-2 \end{pmatrix} g_{n-1} & \begin{pmatrix} n-1 \\ n-1 \end{pmatrix} g_{n} \\ 
0 & -1 & \begin{pmatrix} n-2 \\ 0 \end{pmatrix} g_1 & \cdots & \begin{pmatrix} n-2 \\ n-3  \end{pmatrix} g_{n-2} & \begin{pmatrix} n-2 \\ n-2 \end{pmatrix} g_{n-1} \\
\vdots & \vdots & \vdots & & \vdots & \vdots \\
0 & 0 & 0 & \cdots & \begin{pmatrix} 1 \\ 0 \end{pmatrix} g_1 & \begin{pmatrix} 1 \\ 1 \end{pmatrix} g_2 \\
0 & 0 & 0 & \cdots & -1 & \begin{pmatrix} 0 \\ 0 \end{pmatrix} g_1
\end{pmatrix},
\end{eqnarray}

\noindent and from the determinantal forms of the Bell polynomials at levels $n$ and $n+1$, it is possible to extract the recurrence relation

\begin{eqnarray}
\label{recurrence relation}
Y_{n+1} = g_1 Y_n + \sum_{i=1}^{n} \begin{pmatrix} n \\ i \end{pmatrix} g_{i+1} Y_{n-i}. 
\end{eqnarray}

\noindent Appendix \ref{appendix determinantal form} expands the recurrence relation given in Eq. (\ref{recurrence relation}) up to level 4, showing \textit{en passant} that Eq. (\ref{transformation of Bell polynomial}) holds at those levels. Now for simplicity, using the expression $\tilde{Y}_n= Y\left( \tilde{g}_1, \tilde{g}_2, \ldots, \tilde{g}_n \right)$ while keeping the standard notation $Y_n = Y \left( g_1,g_2, \ldots, g_n \right)$, we obtain

\begin{subequations}
\begin{align}
\tilde{Y}_{n+1} &= \tilde{g}_1 \tilde{Y}_n + \sum_{i=1}^{n} \begin{pmatrix} n \\ i \end{pmatrix} \tilde{g}_{i+1} \tilde{Y}_{n-i} \\ 
&= \left[ \left( q \bar{q} \right) g_1 \right] \left[ \left( q \bar{q} \right)^n Y_n \right] + \sum_{i=1}^{n} \begin{pmatrix} n \\ i \end{pmatrix} \left[ \left( q \bar{q} \right)^{i+1} g_{i+1} \right] \left[ \left( q \bar{q} \right)^{n-i} Y_{n-i} \right] \\
&= \left( q \bar{q} \right)^{n+1} \left[ g_1 Y_n \right] + \sum_{i=1}^{n} \begin{pmatrix} n \\ i \end{pmatrix} \left( q \bar{q} \right)^{i+1+n-i} \left[ g_{i+1} Y_{n-i} \right] \\
&= \left( q \bar{q} \right)^{n+1} \left[ g_1 Y_n \right] + \sum_{i=1}^{n} \begin{pmatrix} n \\ i \end{pmatrix} \left( q \bar{q} \right)^{n+1} \left[ g_{i+1} Y_{n-i} \right] \\
&= \left( q \bar{q} \right)^{n+1} \left[ g_1 Y_n \right] + \left( q \bar{q} \right)^{n+1} \sum_{i=1}^{n} \begin{pmatrix} n \\ i \end{pmatrix}  \left[ g_{i+1} Y_{n-i} \right] \\
&= \left( q \bar{q} \right)^{n+1} \left[ g_1 Y_n  +  \sum_{i=1}^{n} \begin{pmatrix} n \\ i \end{pmatrix}  g_{i+1} Y_{n-i} \right] \\
&= \left( q \bar{q} \right)^{n+1} Y_{n+1}.
\end{align}
\end{subequations}
\end{proof}

After proving that the numerator of the Hilbert series $Z_n \left( q,\bar{q};\mathbb{C}^2 \right)$ is palindromic, we make a more precise statement about the Hilbert series by writing its general formula as 

\begin{eqnarray}
\label{palindromic Zn}
Z_n \left( q,\bar{q};\mathbb{C}^2 \right) = \left[ \prod_{i=1}^n \mathcal{G}_1 \left( q^i,\bar{q}^i \right) \right] P_{m,\bar{m}}( q, \bar{q}) = \frac{P_{m,\bar{m}}( q, \bar{q})}{\prod\limits_{i=1}^n \left( 1-q^i \right) \left( 1-\bar{q}^i \right)} = \frac{\sum\limits_{k=0}^{m} \sum\limits_{\bar{k}=0}^{\bar{m}} a_{k,\bar{k}} q^k \bar{q}^{\bar{k}}}{\prod\limits_{i=1}^n \left( 1-q^i \right) \left( 1-\bar{q}^i \right)}. 
\end{eqnarray}

Appendix \ref{appendix palindromic} gives the expressions of the Hilbert series as in Eq. (\ref{palindromic Zn}) up to order four, showing that at each of those levels, the numerators are indeed palindromic, while the denominators are formed from products of the $\mathcal{G}_1 \left( q^i,\bar{q}^i \right)$ functions. These results allow us to extract an interesting geometric property of the moduli space. Due to the Hochster-Roberts theorem \cite{hochster1974rings} in commutative algebra, the coordinate ring giving the description of the moduli space is Cohen-Macaulay. Indeed, the fundamental theorem of Hochster and Roberts asserts that the invariant ring of a reductive group, $i.e$ a type of linear algebraic group over a field, is Cohen-Macaulay. Among the reductive groups is the general linear group $GL_n$ of invertible matrices. A maximal torus of $GL_n$ is the subgroup of invertible diagonal matrices, whose normalizer is the generalized permutation matrices. The quotient of the normalizer of a maximal torus by the torus is called the Weyl group of a reductive group. In the case of $GL_n$, the Weyl group is the symmetric group $S_n$. Since $S_n$ may be considered as a subgroup of $GL_n$, $S_n$ is a reductive group and the invariant ring $S^n \left( \mathbb{C}^2  \right)$ is therefore Cohen-Macaulay. Furthermore, Stanley's theorem \cite{stanley1978hilbert} states that a Cohen-Macaulay ring that has a palindromic Hilbert series is a Gorenstein ring. The numerator of the Hilbert series $Z_n \left( q,\bar{q};\mathbb{C}^2 \right)$ proved to be palindromic shows that the invariant ring $S^n \left( \mathbb{C}^2  \right)$ is Gorenstein. The importance of this statement is that in the case of affine spaces, Gorenstein means Calabi-Yau. Since our rings are defined over an affine space, we reach the important conclusion that the moduli space is in fact an affine Calabi-Yau orbifold space. In other words, the log sector of the theory lives in the geometry of the $n$-symmetric product of a generic non-compact Calabi-Yau manifold, the two-fold $\mathbb{C}^2$ which has two complex dimensions. The Calabi-Yau orbifold constructed from the gauging of the smooth complex space $\mathbb{C}^2$ by the symmetric group $S_n$ is a conical singular space, as it contains a singularity at the branch point $\mu l=1$, the fixed point of the set of elements of $S_n$ acting on $\mathbb{C}^2$. The $n$-symmetric product of surface $\mathbb{C}^2$ thus constructed plays an important role in understanding the structure of the logarithmic sector of critical massive gravity theories in terms of holomorphic (Riemann) surfaces parametrized in the Calabi-Yau manifold as (covering maps of) $\mathbb{CP}^1$ spaces.

\section{Discussion}
In the sections above, we delved quite intensively into algebro-geometric and invariant theoretic issues that had not yet been addressed in the context of critical massive gravities present in the AdS$_3$/LCFT$_2$ correspondence. In this section, we come to the point where we would like to make some final comments about the logarithmic states and their space geometries. 

Firstly, we would like to mention the following about the "characters" generated by the log-partition $Z_{log} (\nu,q,\bar{q}; \mathbb{C}^2)$. From the above analysis, these objects are linear combinations of the variables $g_1, g_2, \ldots, g_n$ multiplied by a factor $\left( \nu \right)^n$. However, $Z_{log} (\nu,q,\bar{q}; \mathbb{C}^2)$ also seem to describe the symmetric tensor product of the characters of $sl(2,\mathbb{R})$ highest weight representations

\begin{eqnarray}
\label{sl2 characters}
\chi^{sl(2,\mathbb{R})}_h = \frac{q^h}{1-q}, \hspace{1cm} \bar{\chi}^{sl(2,\mathbb{R})}_{\bar{h}} = \frac{\bar{q}^{\bar{h}}}{1-\bar{q}}.
\end{eqnarray}
Taking $h=2,\bar{h}=0$ for instance, it is easy to see that $\chi^{sl(2,\mathbb{R})}_{h=2}$ and $\bar{\chi}^{sl(2,\mathbb{R})}_{\bar{h}=0}$ are respectively the single particle holomorphic and antiholomorphic characters from which $Z_{log} (q^2,q,\bar{q}; \mathbb{C}^2)$ yields symmetric tensor product (multiparticle) characters. Yet, if really there were a holomorphic and an antiholomorphic character, the numerator of the antiholomorphic character of CCTMG should not be constant since the antiholomorphic central charge is not nul. In other words, there should systematically be a $(\bar{\nu})$ appear in the logarithmic partition functions. This argument show that the logarithmic highest weight states have descendants that are both holomorphic and antiholomorphic, just as in the $c=0$ non-unitary LCFT theory \cite{Ridout:2012ew}.

In the extension of holography to the present non-unitary case, critical massive gravities present in the AdS$_3$/LCFT$_2$ correspondence are considered as non-unitary AdS$_3$ holographic duals of two dimensional non-unitary CFTs that are known to exist. We would like to argue that in this setting, the logarithmic states of the critical massive gravities are points on affine Calabi-Yau cones. The idea that conical spaces could play a role in non-unitary holography was recently mentioned in \cite{Raeymaekers:2014kea}, and from the present work, it is quite natural to interpret the conical spaces as affine Calabi-Yau cones. 


\section{Conclusion and outlook}
In this work, we used the partition function derived in \cite{Gaberdiel:2010xv} and reformulated in \cite{Mvondo-She:2018htn} to extract information about the moduli space of the logarithmic sector in critical massive gravities. Using the relation between the cycle index of the symmetric group and Bell polynomials, we first showed that the partition function of the logarithmic states is the generating function of polynomials counting objects invariant under the action of the symmetric group $S_n$. We also showed that the partition function could be expressed as a generating function of 1-part Schur polynomials. It was then shown that the configuration space of the logarithmic states is the symmetric product $S^n(\mathbb{C}^2)$, by showing that the partition function of the logarithmic states $Z_{log} (\nu,q,\bar{q})$ is the generating function of a bi-graded Molien series. As the quotient of an affine space by the reductive group $S_n$, the space $S^n(\mathbb{C}^2)$ has the structure of an orbifold, and $Z_{log} (\nu,q,\bar{q}; \mathbb{C}^2)$ is then the generating function of an $S^n(\mathbb{C}^2)$ orbifold space. Then, the construction of differential operators on orbifolds was discussed. The ring of polynomials on the symmetric product was shown to correspond to bosonic wavefunctions of an $n$-particle system on $\mathbb{C}^2$, that can be mapped tp a Fock space obtained by acting on the Hilbert series with the generators of a Heisenberg-Weyl algebra. Finally, it was proved that the Hilbert series of the polynomial ring have interesting palindromic properties, indicating that the moduli space of the logarithmic states is a Calabi-Yau singular space. Based on this work, we argued that the logarithmic states of the critical massive gravities are points on affine Calabi-Yau cones. 

On its own right, the logarithmic sector of these critical massive gravities looks like an interesting topic of study, and much more seemingly remains to be unraveled from it. For example, it would be interesting to study the modular properties of $Z_{log} (\nu,q,\bar{q}; \mathbb{C}^2)$, and how the study contributes to the modular properties of $Z_{\rm{TMG}}$. This is the object of current investigation. It was also observed that the counting in the multi particle sector of $Z_{log} (\nu,q,\bar{q}; \mathbb{C}^2)$ can only be done correctly on the account of quantum groupoid (Hopf algebroid) coproducts \cite{yannick}. This matter will be discussed in an upcoming publication. 

Lastly, the author would like to mention that the realization that the partition functions of critical massive gravities could be recast in terms of Bell polynomials must be credited to  \cite{Bytsenko:2017byd,Bytsenko:2017czh}. Therefore, inspired by ideas from \cite{Bytsenko:2017byd,Bytsenko:2017czh,Bonora:2016nqc,Bytsenko:2016msy} where partition functions expressed in terms of Bell polynomials and recast into infinite products eventually lead to the construction of quantum group, knots and link invariants, it would be interesting to investigate how  $Z_{log} (\nu,q,\bar{q}; \mathbb{C}^2)$ can be useful in the construction of topological invariants.

\paragraph{Acknowledgements}

Y. M-S is grateful to A. Hanany for communication regarding the plethystic exponential, and Robert de Mello Koch for useful discussions on representation theoretic aspects of this work. He would also like to thank K. Zoubos for proofreading and commenting on earlier versions of this work. 
\clearpage

\begin{appendices}
\numberwithin{equation}{section}
\section{Derivations of plethystic exponential and Bell polynomial forms of $Z_{log} (\nu,q,\bar{q}; \mathbb{C}^2)$ }
In this appendix, we recall the derivation of $Z_{log} (\nu,q,\bar{q}; \mathbb{C}^2)$ as a plethystic exponential and in a Bell polynomial form \cite{yannick}.

Starting from

\begin{eqnarray}
Z_{log} (\nu,q,\bar{q}; \mathbb{C}^2) = \prod_{m=0}^{\infty} \prod_{\bar{m}=0}^{\infty} \frac{1}{1- \nu q^m \bar{q}^{\bar{m}}},
\end{eqnarray}
we first compute the logarithmic function of $Z_{log} (\nu,q,\bar{q}; \mathbb{C}^2)$ as 

\begin{eqnarray}
\log \left[ Z_{log} (\nu,q,\bar{q}; \mathbb{C}^2) \right] = - \sum_{m,\bar{m}\geq 0} \log (1- \nu q^m \bar{q}^{\bar{m}}).
\end{eqnarray}
Using the well know Maclaurin series

\begin{eqnarray}
\log ( 1 - x) = - \sum_{x=1}^{\infty} \frac{x^n}{n},
\end{eqnarray}
we write

\begin{eqnarray}
\log \left[ Z_{log} (\nu,q,\bar{q}; \mathbb{C}^2) \right] = \sum_{m,\bar{m}\geq 0} \sum_{n=1}^{\infty} \frac{\nu ^n}{n} q^{nm} \bar{q}^{n \bar{m}}. 
\end{eqnarray}
Then, using the Maclaurin series of geometric series

\begin{eqnarray}
\frac{1}{1-x} = \sum_{x=0}^{\infty} x^n,
\end{eqnarray}
we write 

\begin{subequations}
\begin{align}
\log \left[ Z_{log} (\nu,q,\bar{q}; \mathbb{C}^2) \right] &= \sum_{m,\bar{m}\geq 0} \sum_{n=1}^{\infty} \frac{\nu ^n}{n} q^{nm} \bar{q}^{n \bar{m}} \\ &= \sum_{n=1}^{\infty} \frac{\nu ^n}{n} \sum_{m,\bar{m}\geq 0} q^{nm} \bar{q}^{n \bar{m}} \\ \label{PE to BP} &= \sum_{n=1}^{\infty} \frac{\nu ^n}{n} \frac{1}{1-q^n} \frac{1}{1-\bar{q}^n}.
\end{align}
\end{subequations}
Finally, exponentiating the above equation yields the plethystic exponential form of $Z_{log} (\nu,q,\bar{q}; \mathbb{C}^2)$

\begin{eqnarray}
Z_{log} (\nu,q,\bar{q}; \mathbb{C}^2) = PE^{\mathcal{B}} \left[ \frac{1}{(1-q)(1-\bar{q})} \right] =  \exp \left( \sum^{\infty}_{n=1} \frac{\left( \nu \right)^n}{n (1-q^n)(1-\bar{q}^n)}   \right).
\end{eqnarray}

To obtain the Bell polynomial version of $Z_{log} (\nu,q,\bar{q}; \mathbb{C}^2)$, it suffices to continue from Eq. (\ref{PE to BP}) as follows

\begin{subequations}
\begin{align}
\log \left[ Z_{log} (\nu,q,\bar{q}; \mathbb{C}^2) \right]  &= \sum_{n=1}^{\infty} \frac{\nu ^n}{n} \frac{1}{1-q^n} \frac{1}{1-\bar{q}^n}   \\ &= \sum_{n=1}^{\infty} \frac{\nu ^n}{n} \frac{1}{\left|1-q^n \right|} \\ &= \sum_{n=1}^{\infty} \frac{\nu ^n}{n!} (n-1)! \frac{1}{\left|1-q^n \right|}.
\end{align}
\end{subequations}
Introducing the function $g_n$ such that 

\begin{eqnarray}
g_n = (n-1)! \frac{1}{\left|1-q^n \right|},
\end{eqnarray}
we get

\begin{eqnarray}
\log \left[ Z_{log} (\nu,q,\bar{q}; \mathbb{C}^2) \right] = \sum_{n=1}^{\infty} \frac{\nu ^n}{n!} g_n.
\end{eqnarray}
Finally, exponentiation the above logarithmic function gives

\begin{subequations}
\begin{align}
Z_{log} (\nu,q,\bar{q}; \mathbb{C}^2) &= \exp \left( \sum_{n=1}^{\infty} \frac{\nu ^n}{n!} g_n   \right) \\ &= \label{BP appendix} \sum_{n=0}^{\infty} \frac{Y_n}{n!} \nu^{n},
\end{align}
\end{subequations}
where in Eq. (\ref{BP appendix}), $Z_{log} (\nu,q,\bar{q}; \mathbb{C}^2)$ is the generating function of the Bell polynomials $Y_n$.

\clearpage

\section{Cycle index of the symmetric group}
\label{appendix cycle index}
In this appendix, the notion of cycle index applied to the symmetric group is reviewed.

\begin{definition}
The symmetric group  $S_n$ defined over a finite set $X$ of $n$ objects, is the group of bijective functions from $X$ to $X$ under the operation of composition, which consists of \textbf{permutations} the $n$ objects.
\end{definition}
The permutations of $S_n$ can be expressed in terms of cycles. For instance, considering the set $X={1,2,3,4,5,6}$, the permutation $\pi = (124)(35)(6)$ tells us that $\pi$ maps 1 to 2, 2 to 4, 3 to 5, and 6 to itself. In this case, $\pi$ consists of 3 disjoint cycles.

\begin{definition}
A \textbf{k-cycle}, or \textbf{cycle of length k}, is a cycle containing k elements.
\end{definition}

\noindent Looking back at the example considered above, $\pi = (124)(35)(6)$ contains a 3-cycle, a 2-cycle, and a 1-cycle. \\

In group theory, the elements of any group may be partitioned into \textbf{conjugacy classes}.

\begin{definition}
In any group G, the elements g and h are \textbf{conjugates} if

\begin{eqnarray*}
g=khk^{-1}
\end{eqnarray*}
for $k \in G$. The set of all elements conjugate to a given g is called the \textbf{conjugacy
class} of g. 
\end{definition}

\noindent Hence, when $S_n$ acts on a set $X$, the cycle decomposition of each $\pi \in S_n$ as product of disjoint cycles is associated to the partitions of the objects in the set. For example, if one considers $S_4$, the partitions of 4 and the corresponding conjugacy classes are 

\begin{subequations}
\begin{align}
(1,1,1,1) &\rightarrow   \left\{(e)\right\} \\
(2,1,1) &\rightarrow \left\{(12), (13), (14), (23), (24), (34) \right\} \\
(2,2) &\rightarrow \left\{ (12)(34), (13)(24), (14)(23)  \right\} \\
(3,1) &\rightarrow \left\{  (123),(132),(124),(142),(134),(143),(234),(243) \right\} \\
(4) &\rightarrow \left\{ (1234),(1432),(1423),(1324),(1342),(1243)  \right\}
\end{align}
\end{subequations}

To keep track of the cycle decomposition of the elements of $S_4$, one can use the cycle index polynomial. Representing each object of the set by a coordinate, $i.e$ 1 by $g_1$, 2 by $g_2$, 3 by $g_3$ and 4 by $g_4$, the cycle index of $S_4$ reads

\begin{eqnarray}
Z(S_4) = \frac{1}{24!} \left( g_1^4 + 6 g_1^2 g_2 + 3 g_2^2 + 8 g_1 g_3 + 6 g^4  \right).
\end{eqnarray}

\noindent The coefficients before the monomials (the products of coordinates) count the number of elements in a given conjugacy class, the powers on the monomials indicate the number of times the object appears in a given partition, and the denominator $24!$ is the order of $S_4$, $i.e$ the total number of elements in $S_4$. As such, the cycle index is simply the average of the number of elements in $X$, that are left invariant by the action of $\pi \in S_4$.

\clearpage 

\section{Some properties of Schur polynomials and Young diagrams}
\label{Properties of Schur polynomials}
We show how to obtain the Schur polynomials in $S_2$ and $S_3$. We start by defining the Schur polynomial as 

\begin{eqnarray}
\label{schur polynomials definition in appendix}
\chi_R(X) =  \frac{1}{n!} \sum_{\pi \in S_n} \chi_R(\pi) \cdot tr \left( \pi X^{\otimes n} \right) = \frac{1}{n!} \sum_{\pi \in S_n}  \chi_R(\pi) \cdot X^{i_1}_{i_{\pi(1)}} X^{i_2}_{i_{\pi(2)}} \cdots X^{i_n}_{i_{\pi(n)}}. 
\end{eqnarray}

 The label $R$ is a Young diagram of $n$ boxes, in one-to-one correspondence with irreducible representations of the symmetric group $S_n$, indicating that Schur polynomials have the property of being associated with a particular irreducible representation of $S_n$. The factor $\chi_R(\pi)$ is the character of $\pi \in S_n$ in the irreducible representation $R$, or in other words the trace of the associated matrix representing $\pi$ in the irreducible representation $R$. $\pi (i)$ represents the integer $i$ that is permuted under the action of the permutation $\pi$. \medskip \\
\indent All group elements with a particular cycle structure belongs to the same conjugacy class, and have the same character for a given irreducible representation. Furthermore, all multi-trace factors are equal for group elements belonging to a particular conjugacy class. Therefore, all permutations of the same cycle structure give the same multi-trace factor. \\ For example, if we consider $n=4$, and take $\pi= (12)(3)(4)$, then we obtain

\begin{eqnarray}
tr \left( (12)(3)(4) \right) = X^{i_1}_{i_2} X^{i_2}_{i_1} X^{i_3}_{i_3} X^{i_4}_{i_4} = tr \left( X^2 \right) \cdot tr (X)^2.
\end{eqnarray}

\subsection{Schur polynomials in $S_2$}
The character table for $S_2$ is 

\begin{center}
\renewcommand{\arraystretch}{1.5}
\begin{tabular}{ |c|c|c c c| } 
 \hline
 Representation & Partition & & Class & \\ 
  & & (1$^2$) & & (2) \\ \hline 
$\Yboxdim{8pt} \yng(2)$ & (2) & 1 & & 1\\ \hline
$\Yboxdim{8pt} \yng(1,1)$ & (1,1) & 1 & & -1\\ [2ex] \hline
\end{tabular}
\end{center}

\noindent $R= \Yboxdim{4pt} \yng(2)$ is known as the symmetric representation, $R= \Yboxdim{4pt} \yng(1,1)$ as the antisymmetric representation. The construction of $\chi_{\Yboxdim{4pt} \yng(2)}$ and $\chi_{\Yboxdim{4pt} \yng(1,1)}$ is done using the character table and Eq. (\ref{schur polynomials definition in appendix}).

\begin{subequations}
\begin{align}
\chi_{\Yboxdim{4pt} \yng(2)}(X) &= \frac{1}{2!} \sum_{\pi \in S_2}  \chi_R(\pi) \cdot 
X^{i_1}_{i_{\pi(1)}} X^{i_2}_{i_{\pi(2)}} \\ 
&= \frac{1}{2!} \left( \chi_R (1) X^{i_1}_{i_1} X^{i_2}_{i_2} + \chi_R (12) X^{i_1}_{i_2} X^{i_2}_{i_1}  \right) \\
&= \frac{1}{2!} \left( (trX)^2 + tr(X^2)  \right).
\end{align}
\end{subequations}

\noindent Similarly

\begin{subequations}
\begin{align}
\chi_{\Yboxdim{4pt} \yng(1,1)}(X) &= \frac{1}{2!} \sum_{\pi \in S_2}  \chi_R(\pi) \cdot 
X^{i_1}_{i_{\pi(1)}} X^{i_2}_{i_{\pi(2)}} \\ 
&= \frac{1}{2!} \left( \chi_R (1) X^{i_1}_{i_1} X^{i_2}_{i_2} + \chi_R (12) X^{i_1}_{i_2} X^{i_2}_{i_1}  \right) \\
&= \frac{1}{2!} \left( (trX)^2 - tr(X^2)  \right).
\end{align}
\end{subequations}

\subsection{Schur polynomials in $S_3$}
The character table for $S_3$ is 

\begin{center}
\renewcommand{\arraystretch}{1.5}
\begin{tabular}{ |c|c|c c c| } 
 \hline
 Representation & Partition & & Class & \\ 
  & & (1$^3$) & (12) & (3) \\ \hline 
 $\Yboxdim{8pt} \yng(3)$ & (3) & 1 & 1 & 1\\ \hline
 $\Yboxdim{8pt} \yng(2,1)$ & (2,1) & 2 & 0 & -1\\ [2ex] \hline
 $\Yboxdim{8pt} \yng(1,1,1)$  & (1,1,1) & 1 & -1 & 1 \\ [3.5ex] \hline 
\end{tabular}
\end{center}

\noindent Here, $R= \Yboxdim{4pt} \yng(3)$ is the symmetric representation, $R= \Yboxdim{4pt} \yng(1,1,1)$ the antisymmetric representation, and $R= \Yboxdim{4pt} \yng(2,1)$ the mixed representation. Again, using the character table and Eq. (\ref{schur polynomials definition in appendix}), we obtain the construction of $\chi_{\Yboxdim{4pt} \yng(3)}$ as follows

\begin{subequations}
\begin{align}
\chi_{\Yboxdim{4pt} \yng(3)}(X) &= \frac{1}{3!} \sum_{\pi \in S_3}  \chi_R(\pi) \cdot 
X^{i_1}_{i_{\pi(1)}} X^{i_2}_{i_{\pi(2)}} X^{i_3}_{i_{\pi(3)}} \\ 
&= \frac{1}{3!} \left( \chi_R (1) X^{i_1}_{i_1} X^{i_2}_{i_2} X^{i_3}_{i_3} + \chi_R (12) X^{i_1}_{i_2} X^{i_2}_{i_1} X^{i_3}_{i_3} +  \chi_R (13) X^{i_1}_{i_3} X^{i_2}_{i_2} X^{i_3}_{i_1} \right. \\ 
&+ \left. \chi_R (23) X^{i_1}_{i_1} X^{i_2}_{3} X^{i_3}_{i_2} + \chi_R (123) X^{i_1}_{i_2} X^{i_2}_{i_3} X^{i_3}_{i_1} + \chi_R (132) X^{i_1}_{i_3} X^{i_2}_{i_1} X^{i_3}_{i_2} \right) \\
&= \frac{1}{3!} \left( \chi_R (1) (tr X)^3 + 3 \chi_R (12) (tr X)(tr X^2) + 2 \chi_R (123)  (tr X^3) \right) \\
&= \frac{1}{3!} \left( (tr X)^3 + 3 (tr X)(tr X^2) + 2 (tr X^3)  \right).               \end{align}
\end{subequations}

\noindent Similarly, 

\begin{subequations}
\begin{align}
 \chi_{\Yboxdim{4pt} \yng(2,1)} &= \frac{1}{3!} \left( \chi_R (1) (tr X)^3 + 3 \chi_R (12) (tr X)(tr X^2) + 2 \chi_R (123)  (tr X^3) \right) \\
&= \frac{1}{3!} \left( 2(tr X)^3 - 2 (tr X^3)  \right),     
\end{align}
\end{subequations}

\noindent and

\begin{subequations}
\begin{align}
\chi_{\Yboxdim{4pt} \yng(1,1,1)} &= \frac{1}{3!} \left( \chi_R (1) (tr X)^3 + 3 \chi_R (12) (tr X)(tr X^2) + 2 \chi_R (123)  (tr X^3) \right) \\
&= \frac{1}{3!} \left( (tr X)^3 - 3 (tr X)(tr X^2) + 2 (tr X^3)  \right).     
\end{align}
\end{subequations}

\clearpage

\section{Invariant theory}
\label{Invariant theory}
Invariant theory appears in the description of moduli spaces whose points parametrize spaces of interest, and is also useful in the construction of Hilbert schemes with associated Hilbert polynomials. In this appendix, we review some basic concepts of invariant theory (a good reference on this topic can be found in \cite{GBBIB3511}).

\subsection{Ring of polynomials}

\begin{definition}
Let $V$ be a complex vector space, and denote the dual vector space by V*= $\left\{ f : V \rightarrow \mathbb{C} \right\}$. The coordinate ring $\mathcal{R}(V)$ of V is the algebra of functions F:  $V \rightarrow \mathbb{C}$ generated by the elements of V*. The elements of $\mathcal{R}(V)$ are called polynomial functions on V.
\end{definition}

For a fixed basis $e_1, e_2, \ldots, e_n$ of $V$, a dual basis of $V*$ can be expressed by the coordinates $x_1,x_2, \ldots, x_n$ such that $x_i( c_1 e_1 + \cdots + c_n e_n) = c_i$. The coordinate ring $\mathcal{R}(V)$ obtained is $\mathbb{C}\left[ x_1,x_2, \ldots, x_n \right]$, the ring of 
polynomials in $n$ variables $f(x_1,x_2, \cdots, x_n)$ with complex coefficients.

\subsection{Invariant rings of the symmetric group}
The fundamental question at the heart of invariant theory is to ask whether the orbits of a group $G$ that acts on a space $V$ can form a space in their own right. In what follows we will consider the case where $G=S_n$.

Let the symmetry group $S_n$ act on the $n$-dimensional complex vector space $V$. The action of $S_n$ on $V$ translates into an action of $S_n$ on the polynomial ring $\mathbb{C}[x_1, \ldots, x_n] := \mathbb{C}[\bm{x}]$. The objective is then to describe the subring of invariant polynomials, which in the present case is denoted $\mathbb{C}[{\bm{x}}]^{S_n}$. According to a theorem of Hilbert, $\mathbb{C}[{\bm{x}}]^{S_n}$ is finitely generated as an algebra. This means that there exist invariants $I_1, \ldots, I_n \in \mathbb{C}[\bf{x}]$ such that $\mathbb{C}[{\bm{x}}]^{S_n}$ consists exactly of polynomials of the invariant ring $\mathbb{C}[{\bm{g}}]= \mathbb{C}[I_1, \ldots, I_n]$. 

We can summarize the results of section \ref{Symmetry and the cycle index} in the following way. The polynomials invariant under action of $S_n$ are precisely the Bell polynomials $Y$ with coordinates ${\bm{g}}=\left( g_1, \ldots, g_n \right)$. In particular, $Y({\bm{g}}) \in \mathbb{C}[{\bm{x}}]^{S_n}$  are uniquely written as polynomial in the $g_1, \ldots, g_n$ such that we have the isomorphism

\begin{eqnarray}
\mathbb{C}[\bm{x}]^{S_n} \sim \mathbb{C}[\bm{g}].
\end{eqnarray}

\subsection{Counting the number of invariants}
In this subappendix, we are interested in counting the polynomials that remain invariant under the action of the symmetric group. The treatment of this enumerative problem can be made systematic by keeping track of the degrees in which these invariants occur.

Let $\mathbb{C}[\bm{x}]^{S_n}_d$ be the set of all homogeneous invariants of degree $d$. The invariant ring $\mathbb{C}[\bm{x}]^{S_n} = \bigoplus_{d=0}^{\infty} \mathbb{C}[\bm{x}]^{S_n}_d$ is the direct sum of the finite dimensional $\mathbb{C}$-vector spaces $\mathbb{C}[\bm{x}]^{S_n}_d$. The Hilbert series (or Poincar\'{e} series) of the graded algebra $\mathbb{C}[\bm{x}]^{S_n}$ is the formal power series in $t$ defined by

\begin{eqnarray}
\label{Hilbert series appendix}
H \left( \mathbb{C}[\bm{x}]^{S_n},t \right) = \sum_{d=0}^{\infty} \mbox{dim} \left( \mathbb{C}[\bm{x}]^{S_n}_d \right) t^d,
\end{eqnarray}
which encodes in a convenient way the dimensions of the $\mathbb{C}[\bm{x}]^{S_n}_d$-vector space of degree $d$.

In 1897, Molien proved that for any group finite group $G$ acting on $\mathbb{C}[\bm{x}]^{G}$, it is possible to compute $H \left( \mathbb{C}[\bm{x}]^{G},t \right)$ without first computing $\mathbb{C}[\bm{x}]^{G}$. This is captured in the beautiful theorem below

\begin{theorem}
(Molien's theorem). Let $\rho: G \rightarrow \rm{GL(V)}$ be a representation of a finite group G of order $\left| G \right|$. If G acts on $\mathbb{C}[V]= \mathbb{C}[\bm{x}]$, then the Hilbert series of the invariant ring $\mathbb{C}[\bm{x}]^{G}$ can be expressed as 

\begin{eqnarray}
H \left( \mathbb{C}[\bm{x}]^{G},t \right) = \frac{1}{|G|} \sum_{g \in G} \frac{1}{\det \left( I - \rho(g) t \right)}.
\end{eqnarray}
\end{theorem}

\noindent We refer the reader to \cite{GBBIB3511} for a very readable proof. In the case of the symmetric group, one simply writes the Hilbert series (\ref{Hilbert series appendix}) as 

\begin{eqnarray}
H \left( \mathbb{C}[\bm{x}]^{S_n},t \right) = \frac{1}{|S_n|} \sum_{g \in S_n} \frac{1}{\det \left( I - \rho(g) t \right)}.
\end{eqnarray}

\subsection{Rings of differential operator}
As the algebra of differential operators on affine $n$-spaces, the \textit{Weyl algebra} is perhaps the most important ring of differential operator. It is denoted 

\begin{eqnarray}
D\left( \mathbb{C}[x_1, \ldots, x_n] \right) = \mathbb{C}<x_1, \ldots, x_n, \partial_{x_1}, \ldots, \partial_{x_n}>,
\end{eqnarray}
where the variables $x_i$ commute with each other, the variables $\partial_{x_j} = \frac{\partial}{\partial x_j}$ commute with each other, and the two sets of variables interact via the product rule $\partial_j x_i = x_i \partial_j + \delta_{ij}$ \cite{traves2006differential}.
\clearpage 
\section{Determinantal form of Bell polynomials} 
\label{appendix determinantal form}

We compute the determinantal form of Bell polynomials at levels 2, 3 and 4, and show from these expressions how the recurrence relation in Eq. (\ref{recurrence relation}) can be seen. \\ \\ \noindent At level 2

\begin{eqnarray}
Y_2 \left( g_1,g_2 \right) = \begin{vmatrix} g_1 & g_2 \\ -1 & g_1 \end{vmatrix} = g_1^2 + g_2 =g_1 Y_1 + g_2 Y_0.
\end{eqnarray}

\noindent At level 3

\begin{subequations}
\begin{align}
Y_3 \left( g_1,g_2,g_3 \right) &= \begin{vmatrix} g_1 & 2g_2 & g_3\\ -1 & g_1 & g_2 \\ 0 & -1 & g_1 \end{vmatrix}\\ &= g_1 \begin{vmatrix} g_1 & g_2 \\ -1 & g_1 \end{vmatrix} - 2g_2 \begin{vmatrix} -1 & g_2 \\ 0 & g_1 \end{vmatrix} + g_3 \begin{vmatrix} -1 & g_1 \\ 0 & -1 \end{vmatrix}  \\ &= g_1 \left( g_1^2 + g_2 \right) + 2 g_2 \left( g_1 \right) + g_3 (1) \\ &= g_1 Y_2 + 2 g_2 Y_1 + g_2 Y_0.
\end{align}
\end{subequations}

\noindent At level 4

\begin{subequations}
\begin{align}
Y_4 \left( g_1,g_2,g_3, g_4 \right) &=    \begin{vmatrix} g_1 & 3g_2 & 3g_3 & g_4 \\ -1 & g_1 & 2 g_2 & g_3 \\ 0 & -1 & g_1 & g_2 \\ 0 & 0 & -1 & g_1 \end{vmatrix} \\ 
&= g_1 \begin{vmatrix} g_1 & 2g_2 & g_3 \\ -1 & g_1 & g_2 \\ 0 & -1 & g_1 \end{vmatrix} + \begin{vmatrix} 3g_2 & 3 g_3 & g_4 \\ -1 & g_1 & g_2 \\ 0 & -1 & g_1 \end{vmatrix} \\ 
&= g_1 \left[ g_1 \left( g_1^2 + g_2 \right) + \left( 2 g_2 g_1 + g_3 \right) \right] + \left[ 3g_2 \left( g_1^2 + g_2 \right) + 3 g_3g_1 + g_4 \right] \\ &= g_1 \left( g_1^3 + 3g_1 g_2 + g^3 \right) + 3 g_2 \left( g_1^2 + g_2 \right) + 3 g_3 \left( g_1 \right) + g_4 (1) \\ 
&= g_1 Y_3 + 3 g_2 Y_2 + 3 g_3 Y_1 + g_4 Y_0.
\end{align}
\end{subequations}

In what follows, we show from the above determinantal forms that the transformation property in Eq. (\ref{transformation of Bell polynomial}) holds at levels 2,3 and 4. \\ \\ \noindent
 At level 2

\begin{subequations}
\begin{align}
\tilde{Y}_2&= \tilde{g}_1 \tilde{Y}_1 + \tilde{g}_2 \tilde{Y}_0  \\ &= \left[ \left(q \bar{q} \right) g_1 \right]  \left[ \left( q \bar{q} \right) Y_1 \right] + \left[ \left( q \bar{q} \right)^2 g_2 \right] \left[ Y_0 \right] \\ &= \left( q \bar{q} \right)^2 \left[g_1 Y_1 + g_2 Y_0 \right] \\ &= \left( q \bar{q} \right)^2 Y_2.
\end{align}
\end{subequations}

\noindent At level 3

\begin{subequations}
\begin{align}
\tilde{Y}_3 &= \tilde{g}_1 \tilde{Y}_2 + 2 \tilde{g}_2 Y_1 + \tilde{g}_3 \tilde{Y}_0 \\ &= \left[ \left( q \bar{q} \right) g_1 \right] \left[ \left( q \bar{q} \right)^2 Y_2 \right] + 2 \left[  \left( q \bar{q} \right)^2 g_2 \right] \left[ \left( q \bar{q} \right) Y_1 \right] + \left[  \left( q \bar{q} \right)^3 g_3 \right] \left[ Y_0 \right] \\ &= \left( q \bar{q} \right)^3 \left[ g_1 Y_2 + 2 g_2 Y_1 + g_3 Y_0 \right] \\ &= \left( q \bar{q} \right)^3 Y_3.
\end{align}
\end{subequations}

\noindent At level 4 

\begin{subequations}
\begin{align}
\tilde{Y}_4 &= \tilde{g}_1 \tilde{Y}_3 + 3 \tilde{g}_2 \tilde{Y}_2 + 3 \tilde{g}_3 \tilde{Y}_1 + \tilde{g}_4 \tilde{Y}_0 \\ &= \left[ \left( q \bar{q} \right) g_1 \right] \left[ \left( q \bar{q} \right)^3 Y_3 \right] + 3 \left[ \left( q \bar{q} \right)^2 g_2 \right] \left[ \left( q \bar{q} \right)^2 Y_2 \right] + 3 \left[ \left( q \bar{q} \right)^3 g_3 \right] \left[ \left( q \bar{q} \right) Y_1 \right] + \left[ \left( q \bar{q} \right)^4 g_4 \right] \left[ Y_0 \right] \\ &= \left( q \bar{q} \right)^4 \left[ g_1 Y_3 + 3 g_2 Y_2 + 3 g_3 Y_1 + g_4 Y_0 \right] \\
&= \left( q \bar{q} \right)^4 Y_4.
\end{align}
\end{subequations}
\clearpage 
\section{Palindromic numerators of $Z_n \left( q,\bar{q}; \mathbb{C}^2 \right)$}
\label{appendix palindromic}

We give expressions of $Z_n \left( q,\bar{q}; \mathbb{C}^2 \right)$ up to $n=4$ showing that the numerators are palindromic.

\begin{subequations}
\begin{align}
Z_2 \left( q,\bar{q}; \mathbb{C}^2 \right) &= \frac{1}{2!} \left[ \frac{1}{\left| 1-q \right|^2} + \frac{1}{\left|1-q^2 \right|}    \right] \\ &= \frac{1}{2!} \left[  \frac{1}{\left| 1-q \right| \left| 1-q \right|} + \frac{1}{\left| 1-q \right| \left| 1+q \right|} \right] \\ &= \frac{1}{2!} \left[ \frac{\left| 1-q \right| + \left| 1+q \right|}{\left| 1-q \right| \left| 1-q \right| \left| 1+q \right|}  \right] \\ &= \frac{1}{2!} \left[ \frac{2 + 2 q \bar{q}}{\left| 1-q \right| \left| 1-q^2 \right|}  \right] \\ &= \frac{1+ q \bar{q}}{\prod\limits_{i=1}^2 \left( 1-q^i \right) \left( 1-\bar{q}^i \right)}.
\end{align}
\end{subequations}

\begin{subequations}
\begin{align}
Z_3 \left( q,\bar{q}; \mathbb{C}^2 \right) &= \frac{1}{3!} \left[ \frac{1}{\left| 1-q \right|^3} + 3 \frac{1}{\left| 1-q \right|} \frac{1}{\left|1-q^2 \right|} + 2 \frac{1}{\left| 1-q^3 \right|}  \right] \\ &= \frac{1}{3!} \left[ \frac{\left| 1-q \right| \left| 1-q^2 \right|\left| 1-q^3 \right|+ 3 \left| 1-q \right|^2 \left| 1-q^3 \right| + 2 \left| 1-q \right|^3 \left| 1-q^2 \right|}{\left| 1-q \right|^3 \left| 1-q^2 \right| \left| 1-q^3 \right|} \right] \\ &= \frac{1}{3!} \left[ \frac{\left| 1+q \right| \left| 1-q^3 \right| + 3 \left| 1-q^3 \right| + 2 \left| 1-q \right|^2 \left| 1+q \right|}{\left| 1-q \right| \left| 1-q^2 \right| \left| 1-q^3 \right|} \right] \\ &= \frac{1}{3!} \left[ \frac{6 + 6q^1\bar{q}^1 + 6q^2\bar{q}^1 + 6 q^1 \bar{q}^2 + 6 q^2 \bar{q}^2 + 6 q^3 \bar{q}^3}{\left| 1-q \right| \left| 1-q^2 \right| \left| 1-q^3 \right|} \right] \\ &=  \frac{1 + q^1\bar{q}^1 + q^2\bar{q}^1 +  q^1 \bar{q}^2 +  q^2 \bar{q}^2 +  q^3 \bar{q}^3}{\prod\limits_{i=1}^3 \left( 1-q^i \right) \left( 1-\bar{q}^i \right)}.
\end{align}
\end{subequations}

The palindromic aspect of the numerator can be seen by rewriting $Z_2 \left( q,\bar{q}; \mathbb{C}^2 \right)$ and $Z_3 \left( q,\bar{q}; \mathbb{C}^2 \right)$ in the following way

\begin{multicols}{2}
\begin{minipage}{0.5 \textwidth} \vspace{0.82cm}
$Z_2 \left( q,\bar{q}; \mathbb{C}^2 \right)= \frac{\begin{tabular}{ccc}
 & \textcolor{red}{1} $q^1 \bar{q}^1$ &  \\ 
\textcolor{red}{0} $q^1 \bar{q}^0$ & & \textcolor{red}{0} $q^0 \bar{q}^1$\\  
  & \textcolor{red}{1} $q^0 \bar{q}^0$ &     
\end{tabular}}{\prod\limits_{1=1}^{2} \left( 1-q^i \right) \left( 1-\bar{q}^i \right)}$,
\end{minipage}
\begin{minipage}{0.5 \textwidth}
$Z_3 \left( q,\bar{q}; \mathbb{C}^2 \right)= \frac{\begin{tabular}{ccccc}
& & \textcolor{red}{1} $q^3 \bar{q}^3$ & & \\ 
 &  & \textcolor{red}{1} $q^2 \bar{q}^2$ &  &\\  
 & \textcolor{red}{1} $q^2 \bar{q}^1$ &  & \textcolor{red}{1} $q^1 \bar{q}^2$  &  \\
&  & \textcolor{red}{1} $q^1 \bar{q}^1$ &  &  \\ 
&  & \textcolor{red}{1} $q^0 \bar{q}^0$ &   &  
\end{tabular}}{\prod\limits_{k=1}^{3} \left| 1-q^k  \right|}$.
\end{minipage}
\end{multicols}

\noindent Then, focusing on the coefficients, we have

\begin{multicols}{2}
\begin{minipage}{0.5 \textwidth} \vspace{0.82cm}
$Z_2 \left( q,\bar{q}; \mathbb{C}^2 \right)= \frac{\begin{tabular}{ c c c }
 & \textcolor{red}{1} &  \\ 
\textcolor{red}{0} & & \textcolor{red}{0} \\  
  & \textcolor{red}{1}  &     
\end{tabular}}{\prod\limits_{k=1}^{2} \left| 1-q^k  \right|}$,
\end{minipage}
\begin{minipage}{0.5 \textwidth}
$Z_3 \left( q,\bar{q}; \mathbb{C}^2 \right)=  \frac{\begin{tabular}{ccccc}
& & \textcolor{red}{1}  & & \\ 
 & \textcolor{blue}{0} & \textcolor{red}{1}  & \textcolor{blue}{0} &\\  
\textcolor{blue}{0} & \textcolor{red}{1}  & \textcolor{blue}{0} & \textcolor{red}{1}   & \textcolor{blue}{0} \\
& \textcolor{blue}{0} & \textcolor{red}{1}  & \textcolor{blue}{0} &  \\ 
&  & \textcolor{red}{1}  &   &  
\end{tabular}}{\prod\limits_{k=1}^{3} \left| 1-q^k  \right|}$,
\end{minipage}
\end{multicols}
 
\noindent and we see the symmetry corresponding to a point reflection through the center. Finally, at level $n=4$, we have 

\begin{center}
$Z_4 \left( q,\bar{q}; \mathbb{C}^2 \right)= \frac{\begin{tabular}{ccccccccccc}
&&&&& \textcolor{red}{1} $q^6 \bar{q}^6$ &&&&&\\ 
&&&&& \textcolor{red}{1} $q^5 \bar{q}^5$ &&&&&\\  
&&&&\textcolor{red}{1} $q^5 \bar{q}^4$&&\textcolor{red}{1} $q^4 \bar{q}^5$&&&&  \\
&&& \textcolor{red}{1} $q^5 \bar{q}^3$ && \textcolor{red}{2} $q^4 \bar{q}^4$ && \textcolor{red}{1} $q^3 \bar{q}^5$ &&&  \\
&&&& \textcolor{red}{1} $q^4 \bar{q}^3$ && \textcolor{red}{1} $q^3 \bar{q}^4$ &&&&  \\
&&& \textcolor{red}{1} $q^4 \bar{q}^2$ && \textcolor{red}{2} $q^3 \bar{q}^3$ && \textcolor{red}{1} $q^2 \bar{q}^4$ &&&   \\
&&&& \textcolor{red}{1} $q^3 \bar{q}^2$ && \textcolor{red}{1} $q^2 \bar{q}^3$ &&&&  \\
&&& \textcolor{red}{1} $q^3 \bar{q}^1$ && \textcolor{red}{2} $q^2 \bar{q}^2$ && \textcolor{red}{1} $q^1 \bar{q}^3$ &&&  \\
&&&&\textcolor{red}{1} $q^2 \bar{q}^1$&&\textcolor{red}{1} $q^1 \bar{q}^2$&&&&  \\
&&&&& \textcolor{red}{1} $q^1 \bar{q}^1$ &&&&& \\ 
&&&&& \textcolor{red}{1} $q^0 \bar{q}^0$ &&&&& \\
\end{tabular}}{\prod\limits_{k=1}^{4} \left| 1-q^k  \right|}$,      
\end{center}

\noindent and taking a closer look at the coefficients by writing 

\begin{center}
$Z_4 \left( q,\bar{q}; \mathbb{C}^2 \right) = \frac{\begin{tabular}{ccccccccccc}
&&&&& \textcolor{red}{1}  &&&&&\\ 
&&&& \textcolor{blue}{0} & \textcolor{red}{1}  & \textcolor{blue}{0} &&&&\\  
&&&\textcolor{blue}{0}&\textcolor{red}{1} &\textcolor{blue}{0}&\textcolor{red}{1} &\textcolor{blue}{0}&&&  \\
&&\textcolor{blue}{0}& \textcolor{red}{1}  &\textcolor{blue}{0}& \textcolor{red}{2}  &\textcolor{blue}{0}& \textcolor{red}{1}  &\textcolor{blue}{0}&&  \\
&\textcolor{blue}{0}&\textcolor{blue}{0}&\textcolor{blue}{0}& \textcolor{red}{1}  &\textcolor{blue}{0}& \textcolor{red}{1}  &\textcolor{blue}{0}&\textcolor{blue}{0}&\textcolor{blue}{0} & \\
\textcolor{blue}{0}&\textcolor{blue}{0}&\textcolor{blue}{0}& \textcolor{red}{1}  &\textcolor{blue}{0}& \textcolor{red}{2}  &\textcolor{blue}{0}& \textcolor{red}{1}  &\textcolor{blue}{0}&  \textcolor{blue}{0}&\textcolor{blue}{0} \\
&\textcolor{blue}{0}&\textcolor{blue}{0}&\textcolor{blue}{0}& \textcolor{red}{1}  &\textcolor{blue}{0}& \textcolor{red}{1} &\textcolor{blue}{0}&\textcolor{blue}{0}&\textcolor{blue}{0} & \\
&&\textcolor{blue}{0}& \textcolor{red}{1}  &\textcolor{blue}{0}& \textcolor{red}{2}  &\textcolor{blue}{0}& \textcolor{red}{1}  &\textcolor{blue}{0}& & \\
&&&\textcolor{blue}{0}&\textcolor{red}{1} &\textcolor{blue}{0}&\textcolor{red}{1} &\textcolor{blue}{0}&&& \\
&&&&\textcolor{blue}{0}& \textcolor{red}{1}  &\textcolor{blue}{0}&&&& \\ 
&&&&& \textcolor{red}{1} &&&&& \\
\end{tabular}}{\prod\limits_{k=1}^{4} \left| 1-q^k  \right|}$,  
\end{center} 

\noindent allows us to see the point reflection symmetry through the center. 


\end{appendices}

\clearpage

\bibliography{main}

\providecommand{\href}[2]{#2}\begingroup\raggedright\endgroup


\end{document}